%% file: main-long-version.tex
\documentclass[12pt]{article}

\usepackage{fullpage}
\usepackage{amsmath}
\usepackage{amsthm}
\usepackage{amsfonts}
\usepackage{mathrsfs}
\usepackage{bm}
\usepackage{graphicx}
\usepackage{tikz}
\usepackage{hyperref}

\newcommand{\symdiff}{\Delta}

\newcommand{\naturals}{\mathbb{N}}

\newcommand{\rationals}{\mathbb{Q}}
\newcommand{\reals}{\mathbb{R}}

\newcommand{\symcirc}{\mathrm{s}}
\newcommand{\symlp}{\mathrm{lp}}
\newcommand{\vars}{\mathrm{w}}
\newcommand{\sizevars}{\mathrm{sw}}

\DeclareMathOperator{\id}{id}
\DeclareMathOperator{\ext}{ext}
\DeclareMathOperator{\ar}{ar}

\DeclareMathOperator{\disjointunion}{{\stackrel{.}{\cup}}}

\newcommand{\Astruct}{\mathbb{A}}

\newcommand{\avector}{\mathbf{a}}
\newcommand{\bvector}{\mathbf{b}}

\newcommand{\xvector}{\mathbf{x}}
\newcommand{\yvector}{\mathbf{y}}
\newcommand{\zvector}{\mathbf{z}}
\newcommand{\svector}{\mathbf{s}}

\newcommand{\transpose}{\mathrm{T}}

\newcommand{\Amatrix}{\mathbf{A}}

\newcommand{\FPC}{\mathrm{FPC}}

\newcommand{\CLogic}{\mathrm{C}}
\newcommand{\FOC}{\mathrm{FOC}}
\newcommand{\FO}{\mathrm{FO}}

\newcommand{\Sym}{\mathrm{Sym}}
\newcommand{\Alt}{\mathrm{Alt}}

\newcommand{\Stab}{\ensuremath{\mathrm{St}}}

\newtheorem{lemma}{Lemma}
\newtheorem{corollary}{Corollary}
\newtheorem{theorem}{Theorem}

\begin{document}

\title{{\bf On the Power of Symmetric Linear Programs}}

\author{Albert Atserias$^1$ \;\;\; Anuj Dawar$^2$ \;\;\; Joanna Ochremiak$^2$ \\
\;\\
Universitat Polit\`ecnica de Catalunya$^1$ \\
University of Cambridge$^2$ \\
\;\\
}

\maketitle

\begin{abstract}
\input{abstract.tex}
\end{abstract}

\newpage

\input{section-1-introduction.tex}

\input{section-2-preliminaries.tex}

\input{section-3-from-circuits-to-lps.tex}

\input{section-4-from-lps-to-fpc.tex}

\input{section-5-results-and-apps.tex}

\input{section-6-concluding-remarks.tex}

\bigskip
\bigskip
\noindent\textbf{Acknowledgments} First author partially funded by
European Research Council (ERC) under the European Union's Horizon
2020 research and innovation programme, grant agreement ERC-2014-CoG
648276 (AUTAR) and MICCIN grant TIN2016-76573-C2-1P (TASSAT3).  The
second author was partially supported by a Fellowship of the Alan
Turing Institute under the EPSRC grant EP/N510129/1.  
The third author funded by the 
European Union's Horizon 2020 research and 
innovation programme under the 
Marie Sk{\l}odowska-Curie grant agreement No 795936.
Some of the work
reported here was initiated at the Simons Institute for the Theory of
Computing during the programme on Logical Structures in Computation in
autumn 2016.  We are grateful to Matthew Anderson for some very
stimulating discussions on this topic.  In particular, he suggested
using Lemma \ref{lem:alt} to prove the existence of supports. We also
thank Wied Pakusa for sharing his manuscript on the FPC-completeness
of linear programming with us.

\bibliographystyle{plain}
\bibliography{biblio.bib}

\end{document}

%% file: abstract.tex
We consider families of symmetric linear programs (LPs) that decide a
property of graphs (or other relational structures) in the sense that,
for each size of graph, there is an LP defining a polyhedral lift that
separates the integer points corresponding to graphs with the property
from those corresponding to graphs without the property.  We show that
this is equivalent, with at most polynomial blow-up in size, to families of
symmetric Boolean circuits with threshold gates.  In particular, when
we consider polynomial-size LPs, the model is equivalent to
definability in a non-uniform version of fixed-point logic with
counting (FPC).  Known upper and lower bounds for FPC apply to the
non-uniform version.  In particular, this implies that the class of
graphs with perfect matchings has polynomial-size symmetric LPs while
we obtain an exponential lower bound for symmetric LPs for the class of
Hamiltonian graphs.  We compare and contrast this with previous
results (Yannakakis 1991) showing that any symmetric LPs for the
matching and TSP polytopes have exponential size.  As an application,
we establish that for random, uniformly distributed graphs,
polynomial-size symmetric LPs are as powerful as general Boolean
circuits.  We illustrate the effect of this on the well-studied planted-clique problem.


%% file: section-1-introduction.tex
\section{Introduction} \label{sec:introduction}
The theory of linear programming is a powerful and widely-used tool for
studying combinatorial optimization problems.  By the same token, the
limitations of such methods are an important object of study in
complexity theory.  A major step in this line of work was the seminal
paper of Yannakakis~\cite{Yannakakis1991} that initiated the study of \emph{symmetric}
linear programs for combinatorial problems.

A polytope in $\reals^n$ is the convex hull of a finite set of
  points in $\reals^n$. Dually, it is the intersection of the finite
  number of half-spaces that define its facets.
Consider a language $S \subseteq \{0,1\}^*$ and let $S_n \subseteq
\{0,1\}^n$ be the collection of strings in $S$ of length $n$.  We can
associate with $S_n$ the polytope $P_n \subseteq \reals^n$ that is the
convex hull of the points $\xvector \in \reals^n$ with $0$-$1$
coordinates that correspond to the strings in $S_n$.  If this polytope
has a succinct representation as a system of linear inequalities, we
can use linear programming methods to optimize linear functions over
$S_n$.  In general, a succinct representation might mean that its size
grows polynomially with $n$.  Thus, the size of the polytope $P_n$, say
measured by the number of its facets, is an important measure of the
complexity of $S$.

In general, even when $P_n$ has a large number of facets, it may admit
a succinct representation as the projection onto $\reals^n$ of a
polytope $Q \subseteq \reals^{n+m}$ of higher dimension.  In this
situation, we call $Q$ a \emph{lift} of $P_n$ and $P_n$ a \emph{shadow} of
$Q$.  This is the
basis for so-called \emph{extended formulations} of combinatorial
optimization problems.  It allows us to optimize over
$S_n$ using linear programs with auxiliary variables.  A classic
example is the convex hull of all strings in $\{0,1\}^n$ of odd
Hamming weight, known as the \emph{parity polytope} which has exponentially many facets but has an
extended formulation using only polynomially many inequalities.  An
interesting feature of many such examples of small extended
formulations is that they are strongly symmetric, i.e., any basic
automorphism of the shadow polytope extends to an automorphism of its
lift. 

Yannakakis~\cite{Yannakakis1991} established lower bounds on the size
of symmetric lifts for the perfect matching polytope and the
travelling salesman polytope.  The \emph{perfect matching polytope} on
$2n$ vertices is the convex hull of points in $\{0,1\}^E$ where $E =
{{[2n]} \choose 2}$ which represent the edge sets of a perfect matching
on $2n$ vertices.  Yannakakis shows that any symmetric lift $Q$ of this polytope
necessarily has a number of facets that is exponential in $n$.
Here ``symmetric'' means that any permutation of the $n$ vertices extends
to an automorphism of $Q$.  This lower bound is then used to show a similar lower
bound for symmetric lifts of the Hamilton cycle polytope (also known
as the travelling salesman polytope).  This is the convex hull of
points in  $\{0,1\}^E$ where $E = {{[n]} \choose 2}$ which are the edge
sets of Hamilton cycles of length $n$.  The conclusion is that any
attempt to solve the travelling salesman problem by representing it as
a linear program in a natural way (i.e.\ respecting the symmetries of
the graph) is doomed to be exponential.  These results launched a long
study of extended formulations of combinatorial problems.  Relatively
recently, exponential lower bounds have been established even without
the assumption of symmetry~\cite{Rothvoss}.

There is another way of representing a language $S \subseteq
\{0,1\}^*$ by a family of polytopes that is also considered by
Yannakakis.  Say that $S_n$ is \emph{recognized} by a polytope $P_n$ if
$S_n \subseteq P_n$ and $\{0,1\}^n \setminus S_n$ is disjoint from $P_n$.  In particular, the convex hull of $S_n$
recognizes $S_n$, but it may well be that there are more succinct
polytopes that also do.  Indeed, Yannakakis shows that for any
language $S$ decidable in polynomial time, there is a family of polynomial-size polytopes
whose shadows recognize $S_n$.  Thus, we cannot expect to prove
exponential lower bounds on such polytopes without separating P from
NP. Note that the assumption of symmetry has been dropped here.  What
can we say about symmetric lifts of polytopes recognizing $S_n$?
Yannakakis does not consider this question and it does not appear to
have been studied in the literature.  This
is the question that we take up in this paper.

We consider families of symmetric polytopes for recognizing classes of
graphs (or other relational structures).  This gives an interesting
contrast with the results of Yannakakis.  Our results show that there
\emph{is} a polynomial-size family of symmetric polytopes whose
shadows recognize the class of graphs that contain a perfect
matching.  On the other hand, there is no  family of symmetric
polytopes of sub-exponential size whose shadows recognize the class of
graphs with a Hamiltonian cycle.

We obtain these specific upper and lower bounds by relating the power of symmetric
linear programs to two other natural models of symmetric computation,
based respectively on logic and circuits.  To be precise, we show that
families of symmetric polytope lifts for recognizing a class of
structures are equivalent to families of \emph{symmetric} Boolean
circuits with threshold gates, in the sense that there are
translations between them with at most a polynomial blow-up in size in
either direction.  This places symmetric linear programs squarely in
the context of a fairly robust notion of symmetric computation that
has recently emerged.  It was shown in~\cite{AndersonDawar2017} that
P-uniform families of symmetric circuits with threshold gates are
equivalent to fixed-point logic with counting ($\FPC$), a well-studied
logic in descriptive complexity theory (see~\cite{DawarSigLog}).

Our translation from circuits to linear programs is based on that
given by Yannakakis, but we need to preserve symmetry and, for
threshold gates, this poses a
significant challenge.  To construct symmetric linear
programs that enforce the values of threshold gates we need a sweeping
generalization of the construction of symmetric lifts of the parity
polytope.  In the other direction, we make a detour through logic.
That is, we show how a family of symmetric polytopes can be translated
into a family of formulas of first-order logic with counting, with the
number of variables and the size being tightly bounded based on the
size of the polytopes.  The translation is based on a support theorem,
which allows us to interpret in
the logic, given a linear program $P$ as advice, a version of $P$ for a particular input
structure.  This then allows us to use the result
of~\cite{Anderson:2015} to the effect that solvability of linear
programs is definable in $\FPC$.  

It is interesting to compare our results with the equivalence between
$\FPC$ and P-uniform symmetric threshold circuits established
in~\cite{AndersonDawar2017}.  Our results are stated for the
non-uniform model and it is not clear that they can be made uniform.
In particular, our translation from linear programs to formulas of
counting logic, while it preserves size, is not necessarily computable
in polynomial time.  It involves symmetry checks that are as hard as
the graph isomorphism problem.  On the other hand, the results
in~\cite{AndersonDawar2017} were stated for polynomial-size families
of circuits and we are able to extend them to sizes up to weakly exponential.
The
translation from circuits to formulas given
in~\cite{AndersonDawar2017} was based on a support theorem proved
there which only worked for circuit sizes bounded by
$O(2^{n^{1/3}})$.  We use a stronger support theorem (proved in
Section~\ref{sec:supports}) which enables us to prove the translation
from families of symmetric linear programs to formulas of counting
logic for sizes up to $O(2^{n^{1-\epsilon}})$ for arbitrarily small $\epsilon$.

The upper and lower bounds for symmetric linear programs that we
obtain (such as for the perfect matching and the Hamilton cycle
problem, respectively) are direct consequences of the equivalence with
non-uniform counting logic.  For instance, it is
known~\cite{Anderson:2015} that perfect matching is definable in
$\FPC$ and it follows that it is recognized by a polynomial-size
family of symmetric polytope lifts. Inexpressibility results for
$\FPC$ are usually established by showing lower bounds on the number
of variables required to express a property in counting logic, and
they yield lower bounds even in the non-uniform setting.  In
particular, we tighten known lower bounds on Hamiltonicity to show
that it cannot be expressed with a sub-linear number of variables
hence with weakly exponential size symmetric polytope lifts. Indeed,
if we use the strongest form of our translation from symmetric
LPs to logic formulas, then the lower bound we get for Hamiltonicity is even
exponential, i.e., of type $2^{\Omega(n)}$ where $n$ is the number of
vertices of the graph.  Similar exponential lower bounds for other
NP-complete problems (such as graph 3-colourability and Boolean
satisfiability) follow from known bounds in counting logic.  Indeed,
exponential lower bounds for some problems in P (such as solving
systems of linear equations over finite fields) also follow.  It
should be noted that this establishes exponential lower bounds also on
symmetric threshold circuits, a problem left open
in~\cite{AndersonDawar2017}, where superpolynomial lower bounds were
established.

Another consequence can be derived from the connection with
$\FPC$.  We know that $\FPC$ can express all polynomial-time
properties of \emph{almost all structures} under a uniform
distribution (see~\cite{HellaKolaitisLuosto1996}).  This can be used to show that $\FPC$ can solve the
planted clique problem if, and only if, the problem is solvable in
polynomial time.  The
planted clique problem is that of distinguishing a random graph from
one in which a clique has been planted.  It is a widely studied
problem in the context of lower bounds on linear programming
methods (see e.g.~\cite{AlonKrivelevichSudakov1998soda,FeigeKrauthgamer2003,Baraketal2016,Hopkinsetal2018}).
It is a consequence of our
results that if this problem can be solved in polynomial time, then it is solvable by polynomial
sized symmetric linear programs.  This is significant because a number
of lower bounds have been established for the planted clique problem
for a variety of models of linear and semidefinite programming,
notably the well-studied Lov\'asz-Schrijver, Sherali-Adams and Lasserre
hierarchies. It is noteworthy that all of these hierarchies yield
symmetric linear or semidefinite programs.  Our
results show that these lower bounds cannot be extended to general
symmetric linear programs without separating P from NP.

In Section~\ref{sec:preliminaries} we establish some preliminary
definitions and notation.  Section~\ref{sec:circuitstoLPs} gives the
translation of circuits to linear programs.  This translation is
carried out for a very general notion of symmetry.  For the reverse
translation, from linear programs to logic given in
Section~\ref{sec:fromLPstocircuits}, we restrict to the natural
symmetries on graphs and relational structures.  The main result and
its consequences, including upper and lower bounds are presented in
Section~\ref{sec:resultsandapps}.

%% file: section-2-preliminaries.tex
\section{Preliminaries} \label{sec:preliminaries}

For a natural number $n \in \naturals$, we write $[n]$ for
$\{1,\ldots,n\}$ with the understanding that~$[0] = \emptyset$.  For
any set $X$, by $\Sym_X$ we denote the \emph{symmetric group on} $X$,
that is, the group of all permutations of the set $X$, and by $\Alt_X$
we denote the \emph{alternating group on} $X$, that is, the group of
all even permutations of the set $X$.  In the special case of $X =
[n]$ we write $\Sym_n$ and $\Alt_n$, respectively. Logarithms are base
$2$ with the convention that $\log(0)=0$.

\subsection{Group actions}

By $H \leq G$ we denote the fact that $H$ is a subgroup of $G$.
If $H \leq G$ then, for any $g \in G$, the subset
$gH = \{ gh : h \in H\}$ of $G$ is called a \emph{coset}
of $H$ in $G$. The number of such cosets is
called the \emph{index}
of $H$ in $G$ and is denoted by $[G : H]$.

Recall that for any group $G$, a $G$-\emph{set} is a set $X$
with an action of the group $G$,
where by an \emph{action}
we mean a mapping $\cdot : G \times X \rightarrow X$
such that for any $\pi,\sigma \in G$ and any $x \in X$ we have
$\pi \cdot ( \sigma \cdot x ) = \pi\sigma \cdot x$ and $\id \cdot x  = x$.
Equivalently, $X$ is a $G$-set if it comes with a homomorphism from
the group $G$ to the symmetric group $\Sym_X$ on $X$.
A \emph{homomorphism} from a $G$-set $X$ to a
$G$-set $Y$ is a function $f$
from $X$ to $Y$ such that for any $\pi \in G$ and $x \in X$,
it holds that $\pi \cdot f(x) = f (\pi \cdot x)$.
For a $G$-set $X$, a \emph{stabilizer} of an element $x \in X$
in $G$ consists of all $\pi \in G$ such that $\pi \cdot x = x$. 
It is easy to see that a stabilizer is a subgroup of $G$. We
sometimes denote it by $G_x$.

For any set $X$, by $|X|$ we denote the number of elements in $X$,
by $\mathcal{P}(X)$ we denote the power set
of $X$ and by $X^n$ we denote the set of $n$-tuples of elements of $X$.
Moreover, if
$n \leq |X|$, by $X^{(n)}$
we denote the set of all $n$-tuples of \emph{distinct} elements of $X$. 
In particular, for $n = 0$, both $X^n$ and $X^{(n)}$
are one-element sets consisting of the empty tuple.
If $X$ is a $G$-set, then the action of $G$
on $X$ induces an action of $G$ on
each of the sets $\mathcal{P}(X)$, $X^n$ and $X^{(n)}$
in the natural way:
for any $\pi \in G$, we have
$\pi \cdot T = \{ \pi \cdot x : x \in T \}$, where $T \subseteq X$,
and $\pi \cdot \svector = (\pi \cdot s_1, \ldots, \pi \cdot s_n)$,
where $\svector = (s_1, \ldots, s_n)$ is a tuple from $X^n$ or
$X^{(n)}$. We refer to the latter group action as the
\emph{componentwise} action of the group $G$.

An action of a group $G$ on a set of indices $U$
defines an action of $G$ on
the set of indexed variables $\{x_u\}_{u \in U}$ in the natural way:
$\pi \cdot x_u = x_{\pi \cdot u}$, for
any $\pi \in G$ and $u \in U$.
This extends to vectors of indexed variables,
as discussed in the paragraph above.
For instance, if the set of indices $[n]^2$
comes with the
componentwise action of the group $\Sym_n$, then
for any $\pi \in \Sym_n$ and
$\xvector = (x_{ij})_{i,j \in [n]}$, we have
$\pi \cdot \xvector = (x_{\pi(i)\pi(j)})_{i,j \in [n]}$.
From now on, in the case of vectors of variables we use the notation
$\xvector^\pi$ instead of $\pi \cdot \xvector$.

For any $G$-set $U$, we define
an action of $G$ on the real vector space $\reals^U$
in the following way.
First, the action of $G$ on the standard basis $\{ e_u \}_{u \in U}$,
where $e_u$ is the vector whose $u$-th coordinate is $1$ and
all other coordinates are $0$, is given by
$\pi \cdot e_u = e_{\pi \cdot u}$, for any $\pi \in G$ and $u \in U$.
This way each $\pi \in G$ defines a mapping from $\{ e_u \}_{u \in U}$
to $\{ e_u \}_{u \in U}$. The action of $G$ on the vector space 
spanned by $\{ e_u \}_{u \in U}$ can be seen as the linear extension
of those mappings:
for any $\pi \in G$ and any real vector $\avector = \sum_{u \in U} a_u e_u$,
we have $\pi \cdot \avector = \sum_{u \in U} a_u (\pi \cdot e_u) =
\sum_{u \in U} a_u e_{\pi \cdot u}$.
For instance, if the set of indices $[n]$
comes with the
natural action of the group $\Sym_n$, then
for any $\pi \in \Sym_n$ and for any vector
$\avector = (a_1, \ldots, a_n)$, we have
$\pi \cdot \avector = \sum_{i \in [n]} a_i e_{\pi(i)} =
\sum_{i \in [n]} a_{\pi^{-1}(i)} e_{i} =
(a_{\pi^{-1}(i)}, \ldots , a_{\pi^{-1}(n)})$.
Here again we use the notation
$\avector^\pi$ instead of $\pi \cdot \avector$. This
notational convention
extends to subsets of real vector spaces: for $P \subseteq \reals^U$
we write $P^\pi$ instead of $\pi \cdot P$.

If a group $G$ acts on a set $U$ and a group
$H$ acts on a set $W$, then the product group $G \times H$
acts on the disjoint union $U \disjointunion W$:
given $\pi \in G$ and $\sigma \in H$,
we have $(\pi,\sigma) \cdot u = \pi \cdot u$,
for $u \in U$, and $(\pi,\sigma) \cdot w = \sigma \cdot w$,
for $w \in W$.
Given such an action of the product group $G \times H$,
of particular interest to us is 
its induced action
on $\reals^U \times \reals^W$ and on sets
of variables indexed by $U \disjointunion W$. 

\subsection{Logic and structures}\label{sec:logic}

A (many-sorted relational) vocabulary consists of a finite set of 
sort symbols and a finite set of relation symbols.
Each relation symbol $R$ comes
with an associated natural number $\ar(R)$ called its \emph{arity}
and with an associated \emph{type} which is a product of $\ar(R)$-many
sort symbols $U_{i_1} \times \ldots \times U_{i_{\ar(R)}}$. The
vocabulary $L_G$ of (directed) graphs is single-sorted
and has one relation symbol $E$
of arity two.  If $L$ is a vocabulary, then an $L$-structure
$\Astruct$ is given by disjoint sets $U_1, \ldots, U_s$,
called \emph{domains}, one for each sort symbol in $L$, and a relation
$R^{\Astruct} \subseteq U_{i_1} \times \ldots \times U_{i_{r}}$
for each $R \in L$ of arity $r$ and type
$U_{i_1} \times \ldots \times U_{i_{r}}$. The relation~$R^{\Astruct}$ is 
called the
interpretation of $R$ in $\Astruct$. Whenever this does not lead to
confusion we use $U$ to denote the domain
associated to the sort symbol $U$.
Moreover, when $\Astruct$ is clear from the context, we omit the
superscript in $R^\Astruct$.  All our structures are finite: their
domains are finite sets. A directed graph is an $L_G$-structure; the
graph is undirected if its interpretation of~$E$ is symmetric and
irreflexive. 

In a logic for a many-sorted vocabulary $L$ the variables are typed,
that is, each different sort has its own set of individual variables.
When an $L$-formula is interpreted on an $L$-structure, the variables
range over the domain of their sort. The atomic $L$-formulas
are equalities between variables of the same type,
and formulas of the form
$R(x_1, \ldots, x_r)$, where $R$ is a relation symbol of arity $r$
in $L$, and $x_1, \ldots, x_r$ are variables of appropriate types.

The class of formulas of first-order logic (FO) is the smallest class
of formulas that contains all atomic formulas
and is closed under negation, conjunction, and existential
quantification.  We consider an extension of first-order logic with
\emph{counting quantifiers}.  For each natural number $q$, we have a
quantifier $\exists^{\geq q}$ where $\Astruct \models \exists^{\geq q}
x \,\phi$ if, and only if, there are at least $q$ distinct elements $a
\in A$ such that $\Astruct \models \phi[a/x]$. While the extension of
first-order logic with counting quantifiers is no more expressive than
$\FO$ itself, the presence of these quantifiers does affect the number
of variables that are necessary to express a query.  Let $\CLogic^k$
denote the $k$-variable fragment of this logic, i.e.\ those formulas
in which no more than $k$ variables appear, free or bound.  $\FPC$ is the
extension of first-order logic with fixed-points and counting.  We do
not give a full definition here as it can be found in standard texts
such as~\cite{OttoHabilitationBook}.  We note that formulas of
$\FPC$ have two sorts of variables, ranging respectively over the
elements of the domain of interpretation and over natural numbers
(restricted to the size of the domain), and allow for terms of the
form $\# x \phi$ which denotes the number of elements that satisfy
$\phi(x)$.  We write $\FOC$ to denote the fragment of $\FPC$ without
fixed-point operators, but where we do allow arithmietic operations
($+$ and $\times$) on the number sort.  The \emph{size} of a formula
is defined as the number of its subformulas. For each formula $\phi$
of $\FPC$ (and, per force, $\FOC$), if the formula uses $k$ variables
then, for every $n$, there is a formula $\theta_n$ of $\CLogic^{2k}$ such
that $\phi$ is equivalent to $\theta_n$ on all structures of size at
most $n$. Moreover, $\theta_n$ has size that is polynomial in the size
of $\phi$, in $k$, and in $n^k$. For more on $\FPC$ and its relation
to the bounded variable fragments $\CLogic^k$ we refer
to~\cite{OttoHabilitationBook}.

Rational numbers $q \in \rationals$ are represented by structures of a
single-sorted vocabulary
$L_{\rationals}$ with three monadic relation symbols and
one binary relation symbol $\leq$. If $q = (-1)^b n/d$, where~$n,d \in
\naturals$ and~$b~\in~\{0,1\}$, then the domain of an
$L_{\rationals}$-structure that represents $q$ is $\{0,\ldots,N\}$
where~$N \in \naturals$ is large enough to represent both the
numerator and denominator with $N$ bits. The binary relation $\leq$ is
interpreted by the natural linear order on $\{0,\ldots,N\}$. The first
of the monadic relation symbols of $L_{\rationals}$ is used to
represent the sign $b$ of $q$ by having it empty if, and only if, $b =
0$. The other two monadic relation symbols of $L_{\rationals}$ are
used to represent the bit positions on which the numerator $n$ and the
denominator $d$ have a one. We use zero denominator to represent $\pm
\infty$.

If $I_1,\ldots,I_d$ denote index sets, tensors $u \in \rationals^{I_1
  \times \cdots \times I_d}$ are represented by many-sorted
  structures, with one sort $\bar{I}$ for each index set $I$ on the list
  $I_1, \ldots, I_d$, and one sort $\bar{B}$ for a
domain $\{0,\ldots,N\}$ of bit positions.
The vocabulary $L_{\mathrm{vec},d}$ of these structures has
a binary relation symbol
$\leq$ for the natural linear order on $\{0,\ldots,N\}$ and three
$d+1$-ary relation symbols $P_s$, $P_n$ and $P_d$ for encoding the
signs and the bits of the numerators and the denominators of the
entries of the tensor. Matrices $\Amatrix \in \rationals^{I \times J}$ and
vectors $\avector \in \rationals^I$ are special cases of
these. Indexed sets of vectors $\{ \avector_i : i \in K \} \subseteq
\rationals^I$ and index sets of rationals $\{ b_i : i \in K \}
\subseteq \rationals$ too.

\subsection{Polytopes, lifts, and shadows}

A polytope is a set of the form $P = \{ \xvector \in \reals^U :
\Amatrix \xvector \leq \bvector \}$, where $U$ and $V$ are abstract
non-empty index sets, $\Amatrix \in \reals^{V \times U}$ is a
constraint matrix, and $\bvector \in \reals^V$ is an offset vector. If
we think of $\xvector = ~(x_u)_{u \in U}$ as a sequence of variables,
then the defining system of inequality constraints $\Amatrix \xvector
\leq \bvector$ is called a \emph{linear program}~(LP) for~$P$. Note
that the defining LPs for polytopes are by no means
unique. Typically~$\Amatrix$ and~$\bvector$ can be chosen to have
rational entries, in which case $P$ is represented by a sequence of
linear constraints $(\gamma_v)_{v \in V}$ of one of its defining LPs
with rational entries; i.e., each~$\gamma_v$ is of the form
$\avector_v^\transpose \xvector \leq b_v$, with $\avector_v \in
\rationals^U$ and $b_v \in \rationals$.  The \emph{size} of such an LP
is $(|U|+1)|V|b$, where $b$ is the maximum number of bits it takes to
write all the numerators and all the denominators of the entries of
the $\avector_v$ and $b_v$ in binary.

If $\xvector \in \reals^U$ and $\yvector_1,\ldots,\yvector_m \in
\reals^U$, and $\alpha_1,\ldots,\alpha_m \in \reals$ are such that
$\xvector=\sum_{i=1}^m \alpha_i \yvector_i$, with~$\alpha_i \geq 0$
and~$\sum_{i=1}^m \alpha_i = 1$, then we say that $\xvector$ is a
convex combination of $\yvector_1,\ldots,\yvector_m$. When $0 <
\alpha_i < 1$ for some $i \in [m]$, the convex combination is called
non-trivial.  The convex hull
$\mathrm{conv}(\yvector_1,\ldots,\yvector_m)$ is the set of all convex
combinations of $\yvector_1,\ldots,\yvector_m$.  A point~$\xvector$ of
a polytope~$P$ is called a \emph{vertex} if it cannot be expressed as a
non-trivial convex combination of any two other points of~$P$.  If~$P$
is a polytope and $P \subseteq \{ \xvector \in \reals^U :
\avector^\transpose \xvector \leq b \}$, then the set~$\{ \xvector \in
P : \avector^\transpose \xvector = b \}$ is called \emph{face} of
$P$. The faces of dimension $0$ are the vertices
of~$P$; the faces of
dimension~$1$ are called \emph{edges}; the faces of dimension
$\mathrm{dim}(P)-1$ are called \emph{facets}, where $\mathrm{dim}(P)$
is the dimension of $P$. Each polytope has finitely many faces of each
dimension; in particular finitely many vertices (see
\cite{Schrijver1986BookReprinted1999}).  A polytope is bounded if and
only if it is the convex hull of its vertices.


If $P \subseteq \reals^U \times \reals^W$ is a polytope, its
projection into $\reals^U$ is the set of points $\xvector \in
\reals^U$ for which there exists a point $\yvector \in \reals^W$ with
$(\xvector,\yvector) \in P$. The projection of a polytope is again a
polytope. If $Q$ is the projection of $P$ into $\reals^U$, then we say
that $Q$ is a \emph{shadow} of~$P$, and that $P$ is a \emph{lift} of
$Q$. If $A,B \subseteq \{0,1\}^U$ are disjoint, then we say that $Q$
is a polytope that separates $A$ from $B$ if $A \subseteq Q$ and $B
\subseteq \reals^U \setminus Q$. We also say that $P$ is a polytope
lift that separates $A$ from $B$. If $Q$ separates $A$ from its
complement $\overline{A} = \{0,1\}^U \setminus A$, then we say
that~$Q$ is a polytope that recognizes $A$, and that $P$ is a polytope
lift that recognizes~$A$. Since no point in $\{0,1\}^U$ is in the
convex hull of any set of points in $\{0,1\}^U$ that does not contain
it, the convex hull of $A \subseteq \{0,1\}^U$ always recognizes
$A$. The converse is not true: a polytope could recognize~$A$ and not
be the convex hull of $A$.

Let $P \subseteq \reals^U$ be given by a sequence of linear
constraints $(\gamma_v)_{v \in V}$. If a group $G$ acts on the
set $U$, then for any $\gamma_v$ of the form
$\avector_v \xvector \leq b_v$, we write $\gamma_v^\pi$ for the
linear constraint
$\avector_v \xvector^\pi \leq b_v$. Note that the sequence
$(\gamma_v^\pi)_{v \in V}$ defines $P^\pi \subseteq \reals^U$,
which is again a
polytope. 
As long as this does not lead to confusion, we identify polytopes with
sequences of linear constraints that represent them.
In particular, if we assume the polytope $P \subseteq \reals^U$
to be represented by the sequence of constraints
$(\gamma_v)_{v \in V}$, then by $P^\pi$ we mean
both the permuted polytope itself, and its representation by
the sequence of constraints $(\gamma^{\pi}_v)_{v \in V}$.

Let $U$ be a $G$-set.
A polytope $P \subseteq \reals^U \times \reals^W$
is said to be $G$-\emph{symmetric} if for every
$\pi \in G$ there exists a permutation $\sigma \in \Sym_W$
such that $ P^{(\pi,\sigma)} = P$.
If additionally we are given an action of the group $G$ on $W$
such that $P^{(\pi,\pi)} = P$,
then we say that the polytope $P$ is $G$-symmetric
\emph{with respect to this~action}.
A pair of
permutations $(\pi,\sigma) \in \Sym_U \times \Sym_W$ such
that $P^{(\pi, \sigma)} = P$ is called an \emph{automorphism} of $P$.
Hence, if $G \leq \Sym_U$, the fact that the polytope $P$
is $G$-symmetric is equivalent to the possibility of extending every permutation $\pi \in G$ to an automorphism of $P$.

For any $n \in \naturals$, if the set $[n]^2$
comes with the natural action of the symmetric group $\Sym_n$,
then any $\Sym_n$-symmetric polytope
$P \subseteq \reals^{[n]^2} \times \reals^W$ is said to be
\emph{graph-symmetric}. It is not difficult to see that any set
$A \subseteq \{0,1\}^{[n]^2}$ recognised by a graph-symmetric polytope
lift $P \subseteq \reals^{[n]^2} \times \reals^W$ is invariant with respect
to the action of the group $\Sym_n$, i.e., for any $\avector \in A$ and
any $\pi \in \Sym_n$, we have $\avector^\pi \in A$. The polytope lift
$P$ can be therefore seen as recognising a class of graphs with
$n$-vertices. If we take a graph $G$ with the set of vertices $V$
of size $n$, fix
a bijection $f$ from $[n]$ to $V$, and define
a vector $\avector = (a_{ij})_{i,j \in [n]} \in \{0,1\}^{[n]^2}$ by: 
$a_{ij} = 1$ if and only if there is an edge from $f(i)$ to $f(j)$ in $G$,
then $G$ belongs to the class recognised by $P$ if and only if
$\avector \in A$. Since $A$ is a $\Sym_n$-set, this
definition does not depend on the choice of~$f$.

More generally, we consider $\Sym_n$-symmetric
polytope lifts recognising properties of arbitrary $L$-structures.
For any $n \in \naturals$ and any
single-sorted vocabulary $L$, 
let $L(n)$ be the disjoint union of $[n]^{\ar(R)}$
over all relation symbols $R$ in $L$. Since $L(n)$ comes with the natural
action of the group $\Sym(n)$, we can talk about
$\Sym_n$-symmetric polytopes over $\reals^{L(n)} \times \reals^W$.
Any such polytope is called $L$-\emph{symmetric}. As a 
straightforward generalisation of the discussion in the previous paragraph,
$L$-symmetric polytope lifts over 
$\reals^{L(n)} \times \reals^W$
are defined to recognise classes of $L$-structures
with $n$-element domains.

\subsection{Boolean circuits}

A circuit with inputs $(x_u)_{u \in U}$ is a directed acyclic graph
$G$ whose vertices of zero in-degree are labelled by some input $x_u$,
and every other vertex is labelled by a function from some fixed basis
of symmetric Boolean functions, with the constraint that the function
takes the same number of inputs as the in-degree of the vertex. The
vertices of zero out-degree are called outputs. If $(y_v)_{v \in V}$
is a fixed naming of the outputs, then a circuit computes a Boolean
function $f : \{0,1\}^U \rightarrow \{0,1\}^V$ in the obvious
way. When $m = 1$ we say that it recognizes $f^{-1}(1)$. The vertices
of a circuit are also called \emph{gates}. The size of the circuit is
its number of gates. A \emph{Boolean threshold circuit} is one whose
gates are labelled by NOTs, unbounded degree ANDs, unbounded degree
ORs, or unbounded degree thresholds $\mathrm{TH}_{n,k}$, where
$\mathrm{TH}_{n,k}(z_1,\ldots,z_n)$ outputs $1$ if, and only if, the
number of $1$'s in the input $z_1,\ldots,z_n$ is at least $k$.

If $U$ is a $G$-set and $W$ denotes the set of gates of $C$, we say
that $C$ is \emph{$G$-symmetric} if for every $\pi \in G$ there exists
$\sigma \in \Sym_W$ such that $C^{(\pi,\sigma)}=C$,
where by $C^{(\pi,\sigma)}$ we mean that
the gates of the circuit are permuted according to $\sigma$,
the labels from
$\{x_u\}_{u \in U}$ are permuted according to $\pi$,
and none of the other labels is moved.
A circuit with $U
= L(n)$ is called $L$-symmetric if it is $\Sym_n$-symmetric, with the
natural action of $\Sym_n$ on $L(n)$. As for polytopes, we consider
$L$-symmetric circuits as recognizing classes of $L$-structures on
abstract sets $V$ of vertices through bijections $f : [n] \rightarrow
V$. In the case of graphs, for example, in which $L(n) = [n]^2$,
we say that such
a circuit accepts a graph $G$ with the
set of vertices $V$ of size $n$ if for some, and hence every,
bijection $f : [n] \rightarrow V$ it holds that $C(\avector)=1$,
where~$\avector$ is the vector that describes the image of $G$ under
$f^{-1}$, as in the previous section.


%% file: section-3-from-circuits-to-lps.tex
\section{From Circuits to LPs} \label{sec:circuitstoLPs}

In this section we prove the half of the equivalence that takes
symmetric circuits with threshold gates into symmetric LPs. That is:

\begin{lemma} \label{lem:circuitstolps}
If $\mathscr{C}$ is a class of $L$-structures that is recognized by a
family of $L$-symmetric Boolean threshold circuits of size $s(n)$,
then $\mathscr{C}$ is recognized by a family of $L$-symmetric LP lifts
of size $s(n)^{O(1)}$. In addition, if the Boolean circuits do not
have threshold gates, then the size of the LP lifts is $O(s(n))$.
\end{lemma}

The main step in the construction is the simulation of the threshold
gates. The na\"{\i}ve approach by which each threshold gate is replaced by
an equivalent AND-OR-NOT circuit will not work: it is known that any
symmetric such circuit that computes the majority function must have
superpolynomial size. This follows from Theorem~2
in~\cite{AndersonDawar2017} and a standard Ehrenfeucht-Fraiss\'e
argument.  We need an alternative approach. As a step towards our goal,
first we need to generalize the so-called Parity Polytope from
Yannakakis~\cite{Yannakakis1991}.

\subsection{The parity polytope explained}

Yannakakis gives a polynomial-size symmetric polytope lift of the
parity polytope
\begin{equation}
\mathrm{PP}(n) := \mathrm{conv} \{
(x_1,\ldots,x_n) \in \{0,1\}^n : \textstyle{\sum_{k=1}^n x_k \equiv 1
\;(\text{mod}\; 2)} \}.
\end{equation} 
Note that this polytope has the following
interesting feature: 
\begin{equation}
\text{\parbox{.85\textwidth}{
if $x_1,\ldots,x_{n-1}$ are in $\{0,1\}$, then there exists a unique
$x_n$ in $\reals$ such that $(x_1,\ldots,x_n)$ is in $\mathrm{PP}(n)$,
and moreover this $x_n$ is the unique bit in $\{0,1\}$ that makes
the total sum $\textstyle{\sum_{k=1}^n x_k}$ odd.
}}
\end{equation}
For the existence just take $x_n \in \{0,1\}$ so that $\sum_{k=1}^n
x_k$ is odd. The uniqueness will follow once we show that any $x_n$
for which the extension vector $(x_1,\ldots,x_n)$ belongs to
$\mathrm{PP}(n)$ is in $\{0,1\}$. In turn, this follows from the fact
that all extreme points of $\mathrm{PP}(n)$ are in $\{0,1\}^n$, all have
the same parity, and a single bit-flip flips
their parity. Indeed, if $(x_1,\ldots,x_n)$ is in $\mathrm{PP}(n)$ but
is not an extreme point, then it must be a non-trivial convex
combination of at least two extreme points and, whenever
$x_1,\ldots,x_{n-1}$ are all in $\{0,1\}$, only two candidate extreme
points remain: $(x_1,\ldots,x_{n-1},0)$ and
$(x_1,\ldots,x_{n-1},1)$. Otherwise some $x_i$ with $1 \leq i \leq
n-1$ would be strictly between $0$ and $1$. However, among these two
candidate extreme points, at least one does not have the right parity,
and hence is not even in $\mathrm{PP}(n)$.

The construction of the lift of $\mathrm{PP}(n)$ relies on the fact
that the convex hulls of the Hamming-weight slices of the
$n$-dimensional Boolean hypercube are definable by a small linear
program.  Precisely, for each $t \in [n]$ let
\begin{equation}
\mathrm{EX}(n,t) := \mathrm{conv} \{ (x_1,\ldots,x_n) \in \{0,1\}^n :
\textstyle{\sum_{k=1}^n x_k = t} \}.
\end{equation}
Then it holds that the direct relaxation is tight:
\begin{equation}
\mathrm{EX}(n,t) = \{ (x_1,\ldots,x_n) \in \reals^n:
\textstyle{\sum_{k=1}^n x_k = t}, \; 0 \leq x_k \leq 1\; \text{ for
  every } k = 1,\ldots,n \}. \label{eq:eq}
\end{equation}
We provide an $\pm \epsilon$-proof since it is instructive.   Let $P$ be the polytope defined by the
linear program in the right-hand side of~\eqref{eq:eq}. The inclusion
$\mathrm{EX}(n,t) \subseteq P$ is obvious. For the inclusion $P
\subseteq \mathrm{EX}(n,t)$ it suffices to show that every vertex of
$P$ has integral components; i.e., that every point in $P$ that has
some non-integral component fails to be a vertex of $P$ because it is
a non-trivial convex combination of two other points in $P$. Let
$\xvector = (x_1,\ldots,x_n)$ be a point in $P$, let $I$ be the set of
indices $k \in [n]$ such that $0 < x_k < 1$, and assume that~$I \not=
\emptyset$. Now define
\begin{equation}
\epsilon :=
\min \{ x_k : k \in I \} \cup \{ 1-x_k : k \in I \}.
\end{equation}
Thus $\epsilon > 0$. Since $\sum_{k=1}^n x_k = t$ and $t$ is an
integer, necessarily $|I| \geq 2$. Fix two different indices $i$ and
$j$ in $I$ and, for $b \in \{0,1\}$, let $\xvector_b =
(x_{b,1},\ldots,x_{b,n})$ be the point defined by
\begin{equation*}
\begin{array}{lllll}
x_{b,k} & := & x_k + (-1)^b \epsilon & \;\;\; & \text{for } k = i \\
x_{b,k} & := & x_k - (-1)^b \epsilon & & \text{for } k = j \\
x_{b,k} & := & x_k & & \text{for } k \in [n]\setminus\{i,j\}.
\end{array}
\end{equation*}
By the choice of $\epsilon$ we have $0 \leq x_{b,k} \leq 1$ for
every $k \in [n]$ and every $b \in \{0,1\}$. Also
\begin{equation}
\sum_{k=1}^n x_{b,k} = \sum_{k=1}^n x_k + \epsilon - \epsilon =
\sum_{k=1}^n x_k = t
\end{equation} 
for both $b \in \{ 0,1 \}$, and 
\begin{equation}
\frac{1}{2} (x_{0,k} + x_{1,k}) = \frac{1}{2} (x_k + x_k) = x_k
\end{equation}
for every $k \in [n]$.  Since $\epsilon > 0$, it follows that
$\xvector_0$ and $\xvector_1$ are distinct points in $P$ such that
$\xvector = \frac{1}{2}(\xvector_0 + \xvector_1)$. Thus $\xvector$ is
not a vertex of $P$.

With the polytope $\mathrm{EX}(n,t)$ in hand we are ready to describe
the lift of $\mathrm{PP}(n)$.  First note that a real vector $\xvector
\in \reals^n$ is in $\mathrm{PP}(n)$ if, and only if, $\xvector =
\sum_{t} w_t \yvector_t$ where each vector $\yvector_t$ is in
$\mathrm{EX}(n,2t+1)$, and the $w_t$'s are non-negative coefficients
that add up to one, with $t$ ranging over
$\{0,\ldots,\lfloor{n/2}\rfloor\}$. In order to express this as a linear
program we introduce variables $w_t$ and $z_{t,i}$ for each $t \in 
\{0,\ldots,\lfloor{n/2}\rfloor\}$
and $i \in \{1,\ldots,n\}$ with the intention that $z_{t,i} = w_t
y_{t,i}$ for appropriate values $y_{t,i}$ that we do not care to actually
get. Writing $T$ for $\{0,\ldots,\lfloor{n/2}\rfloor\}$ and $N$ for
$\{1,\ldots,n\}$, the linear program that achieves this is the
following:
\begin{equation}
\begin{array}{lll}
\sum_{t \in T} w_t = 1 & & \\
0 \leq w_t \leq 1 & & \text{for each } t \in T \\
\sum_{t \in T} z_{t,i} = x_i & \;\;\; & \text{for each } i \in N \\
\sum_{i \in N} z_{t,i} = (2t+1) w_t & & \text{for each } t \in T \\
0 \leq z_{t,i} \leq w_t & & \text{for each } t \in T \text{ and } i \in N \\
\end{array}
\end{equation}

\noindent The symmetry of this linear program with respect to the
$x$-variables is obvious: given a permutation $\pi \in \Sym_n$, let
$\sigma$ be the permutation that maps $z_{t,i}$ to $z_{t,\pi(i)}$ and
leaves each $w_t$ in place.

\subsection{The truncated parity polytope}

For each integer $n \geq 1$, let $|n|$ be the number of bits it takes
to write $n$ in binary.  We want lifts of the following
\emph{truncated parity polytopes} defined for each pair of integers~$n
\geq 1$ and~$q \in \{0,\ldots,|n|-1\}$:
\begin{equation}
\mathrm{PP}(n,q) := \mathrm{conv} \{ (x_1,\ldots,x_n) \in \{0,1\}^n :
\textstyle{\lfloor{2^{-q} \sum_{k=1}^n x_k}\rfloor \equiv 1
\;(\text{mod}\; 2)} \}.
\end{equation}
This time a vector $\xvector \in \reals^n$ is in $\mathrm{PP}(n,q)$ if,
and only if, $\xvector = \sum_{t,r} w_{t,r} \yvector_{t,r}$ where each
vector~$\yvector_{t,r}$ is in $\mathrm{EX}(n,2^q (2t+1) + r)$, and the
$w_{t,r}$ are non-negative coefficients that add up to one, with
$(t,r)$ ranging over the set of pairs of integers with $t \in
\{0,\ldots,\lfloor{n/2^{q+1}}\rfloor\}$ and $r \in
\{0,\ldots,2^q-1\}$. We introduce variables $w_{t,r}$ and $z_{t,r,i}$
for each $(t,r) \in T$ and each $i \in N$, where $T =
\{0,\ldots,\lfloor{n/2^{q+1}}\rfloor \} \times \{0,\ldots,2^q-1\}$ and
$N = \{1,\ldots,n\}$, with the intention that $z_{t,r,i} = w_{t,r}
y_{t,r,i}$ for appropriate values $y_{t,r,i}$ that we do not care to
actually get. The linear program that achieves this is the following:
\begin{equation}
\begin{array}{lll}
\sum_{(t,r) \in T} w_{t,r} = 1 & & \\
0 \leq w_{t,r} \leq 1 & & \text{for each } (t,r) \in T \\
\sum_{(t,r) \in T} z_{t,r,i} = x_i & \;\;\; & \text{for each } i \in N \\
\sum_{i \in N} z_{t,r,i} = (2^q(2t+1) + r) w_{t,r} & & \text{for each } (t,r) \in T \\
0 \leq z_{t,r,i} \leq w_{t,r} & & \text{for each } t \in T \text{ and } i \in N \\
\end{array}
\end{equation}
The symmetry of this linear program with respect to the $x$-variables
is again obvious: given $\pi \in \Sym_n$, let $\sigma$ map $z_{t,r,i}$
to $z_{t,r,\pi(i)}$ and leave every other variable in place.  The
polytope~$\mathrm{PP}(n,q)$ has the following interesting feature that
is analogous to the one we argued for~$\mathrm{PP}(n)$:
\begin{equation} 
\text{\parbox{.85\textwidth}{
if $x_1,\ldots,x_{n-1}$ are in $\{0,1\}$ and $\textstyle{\sum_{k=1}^{n-1} x_k
\equiv -1\;(\text{mod}\; 2^q)}$, then there exists a unique $x_n$ in
$\reals$ such that $(x_1,\ldots,x_n)$ is in $\mathrm{PP}(n,q)$, and
moreover $x_n$ is the unique bit in $\{0,1\}$
that makes the truncation $\textstyle{\lfloor{2^{-q} \sum_{k=1}^n x_k}\rfloor}$
odd.
}}
\end{equation}
For the existence just take $x_n \in \{0,1\}$ so that the truncation
$\lfloor{2^{-q} \sum_{k=1}^n x_k}\rfloor$ is odd, which must exist by
the assumption that $\sum_{k=1}^{n-1} x_k \equiv -1\; (\text{mod}\;
2^q)$. The uniqueness follows once we show that every $x_n$ for
which the extension vector $(x_1,\ldots,x_n)$ belongs to
$\mathrm{PP}(n,q)$ is in~$\{0,1\}$, and again the assumption that
$\sum_{k=1}^{n-1} x_k \equiv -1\;(\text{mod}\; 2^q)$.  For a proof
that such an $x_n$ is in~$\{0,1\}$ it suffices to show that if
$(x_1,\ldots,x_n) \in \mathrm{PP}(n,q)$ satisfies the conditions, then
it is an extreme point. If it were not an extreme point then
it would be a non-trivial combination of at least two extreme points
and, whenever $x_1,\ldots,x_{n-1}$ are all in $\{0,1\}$, only two
candidates remain: $(x_1,\ldots,x_{n-1},0)$ and
$(x_1,\ldots,x_{n-1},1)$; otherwise some $x_i$ with $1 \leq i \leq
n-1$ would be strictly between $0$ and $1$. However, assuming that
$\sum_{k=1}^{n-1} x_k \equiv -1\;(\text{mod}\; 2^q)$, at least one of
these extreme points must have even truncation $\lfloor{2^{-q}
  \sum_{k=1}^n x_k}\rfloor$, and hence not even belong to
$\mathrm{PP}(n,q)$; a contradiction.

\subsection{Counting gates}

The goal in this subsection is to write a polynomial-size symmetric
linear program that can be used to simulate exact counting gates
$\mathrm{EX}_{n,t}(x_1,\ldots,x_n)$, which outputs $1$ if
the sum of the~$n$ input bits is exactly $t$, and $0$ otherwise. In
order to do this we use the truncated parity polytopes to compute the
bits of the binary representation of $\sum_{k=1}^n x_i$, and then
compare the result with the bits of the binary representation of $t$.

First consider the following sequence of linear programs which
depend only on $n$ and not~on~$t$:
\begin{equation}
\begin{array}{ll}
& (x_1,\ldots,x_n,1,1-z_1) \in \mathrm{PP}(n+2,0) \\ 
& (x_1,\ldots,x_n,1,1,z_1,1-z_2) \in \mathrm{PP}(n+4,1) \\ 
& (x_1,\ldots,x_n,1,1,1,1,z_1,z_2,z_2,1-z_3) \in \mathrm{PP}(n+8,2) \\ 
& \vdots \\ 
& (x_1,\ldots,x_n,1^{(2^q)},z_1^{(1)},z_2^{(2)},\ldots,z_q^{(2^{q-1})},
  1-z_{q+1}) \in \mathrm{PP}(n+2^{q+1},q) \\ 
& \vdots \\ 
& (x_1,\ldots,x_n,1^{(2^{|n|-1})},z_1^{(1)},z_2^{(2)},\ldots,z_{|n|-1}^{(2^{|n|-2})},
  1-z_{|n|}) \in \mathrm{PP}(n+2^{|n|},|n|-1),
\end{array}
\label{eqn:firstpart}
\end{equation}
where, for $\ell \geq 1$, the notation $a^{(\ell)}$ denotes the string
$a,a,\ldots,a$ of length $\ell$.  We claim the following property:
\begin{equation}
\text{\parbox{.85\textwidth}{
  if $x_1,\ldots,x_n$ are in $\{0,1\}$, then there is a unique vector
  $(z_1,\ldots,z_{|n|}) \in \reals^{|n|}$ which together with
  $x_1,\ldots,x_n$ is a solution, and in
  this solution $z_k \in \{0,1\}$ for all $k$, and 
  $\textstyle{\sum_{k=1}^n x_k = \sum_{k=1}^{|n|} (1-z_k) 2^{k-1}}$;
  in other words,
  $z_1,\ldots,z_{|n|}$ are the flips of the bits of the binary
  representation of $\textstyle{\sum_{k=1}^n x_k}$, listed from least to most
  significant bit.
}} \label{prop:truncatedpp}
\end{equation}
From now on in this proof, let $X = \sum_{k=1}^n x_k$.  The first part
of the claim follows from the corresponding property of
$\mathrm{PP}(n,q)$'s, and induction on $q = 0,1,\ldots,|n|-1$.  The
second part is proved by showing that $z_1,\ldots,z_{q+1}$ are all in $\{0,1\}$
and
\begin{equation}
X \equiv \sum_{k=1}^{q+1} (1-z_k) 2^{k-1}\; (\text{mod}\; 2^{q+1})
\label{eq:ih}
\end{equation}
also by induction on $q = 0,1,\ldots,|n|-1$.  For $q = 0$ the claim is
true since, by the main property of $\mathrm{PP}(n+2,0) =
\mathrm{PP}(n+2)$, there is a unique $z_1$ for which
$(x_1,\ldots,x_n,1,1-z_1)$ is in $\mathrm{PP}(n+2)$, and this is the
unique bit that makes $X + 1 + 1 - z_1$ odd, hence $X - z_1$ odd.
Assume now that~$z_1,\ldots,z_{q+1}$ are all in $\{0,1\}$ and
that~\eqref{eq:ih} holds for $q \in \{0,\ldots,|n|-2\}$ and we prove
it for $q+1$. First observe that
\begin{equation}
X + 2^{q+1} + \sum_{k=1}^{q+1} z_k 2^{k-1}
\equiv X + \sum_{k=1}^{q+1} z_k 2^{k-1} \equiv 2^{q+1} - 1 \equiv -1 \;(\text{mod}\; 2^{q+1}),
\end{equation}
where in the second equivalence we used the induction hypothesis on $q$.
Moreover $z_1,\ldots,z_{q+1}$ are all in $\{0,1\}$, also by induction hypothesis on $q$.
Thus, the vector
\begin{equation}
  (x_1,\ldots,x_n,1^{(2^{q+1})},z_1^{(1)},z_2^{(2)},\ldots,z_{q+1}^{(2^{{q+1}-1})},1-z_{q+2})
\label{eqn:longgg}
\end{equation}
satisfies the hypothesis of property~\eqref{prop:truncatedpp} for
$\mathrm{PP}(n+2^{q+2},q+1)$. It follows that there is a
unique~$z_{q+2}$ that puts~\eqref{eqn:longgg}
in~$\mathrm{PP}(n+2^{q+2},q+1)$, and this is the unique bit such that
the quantity
\begin{equation}
\lfloor{{2^{-(q+1)} (X + 2^{q+1} + \sum_{k=1}^{q+1} z_k 2^{k-1} + 1 - z_{q+2})}}
\rfloor \label{eq:thirteen}
\end{equation}
is odd. Now,~\eqref{eq:ih} says that
\begin{equation}
X = 2^{q+1}\lfloor{2^{-(q+1)}X}\rfloor + 
\sum_{k=1}^{q+1} (1-z_k) 2^{k-1}, 
\end{equation}
which means that the unique bit that makes the quantity in~\eqref{eq:thirteen}
odd is the flip of the $(q+1)$-th least significant bit of $X$.

Now, exact-$t$ counting gates can be expressed
using an additional linear program that simulates an AND gate to
compare the bits $z_1,\ldots,z_{|n|}$ with (the flips of) the bits of
the binary representation of $t$. For both $b \in
\{0,1\}$, let $K_b \subseteq [|n|]$ be the subset of bit-positions at
which the $|n|$-long binary representation of $t$ is $b$. Then, the
relation $y = \mathrm{EX}_{n,t}(x_1,\ldots,x_n)$ is expressed by the
union of~\eqref{eqn:firstpart} and the following:
\begin{equation}
\begin{array}{lll}
y \geq \sum_{k \in K_0} z_k + \sum_{k \in K_1} (1-z_k) - |n| + 1 & \;\;\; & \\
y \leq z_k & \;\;\; & \text{for every } k \in K_0 \\
y \leq 1-z_k & & \text{for every } k \in K_1 \\
0 \leq y \leq 1. 
\end{array}
\label{eqn:secondpart}
\end{equation}
We write $\mathrm{EX}_{n,t}(x_1,\ldots,x_n, y)$ to denote
the LP that has all the constraints in~\eqref{eqn:firstpart} and all
the constraints in~\eqref{eqn:secondpart}. We summarize its main
properties in the following:

\begin{lemma} \label{lem:exgates}
The linear program $\mathrm{EX}_{n,t}(x_1,\ldots,x_n, y)$ has size
polynomial in $n$, is symmetric with respect to the group of
permutations of $x_1,\ldots,x_n,y$ that fix $y$, and has the following
property: If $x_1,\ldots,x_n$ are all in $\{0,1\}$, then there is a
unique $y \in \reals$ such that $(x_1,\ldots,x_n, y)$ can be extended
to a feasible solution, and this $y$ is the unique output bit of the
corresponding gate evaluated on inputs $x_1,\ldots,x_n$.
\end{lemma}

\begin{proof}
The bound on the size follows by inspection. The symmetry with respect
to the group of permutations of $x_1,\ldots,x_n,y$ that fix $y$
follows from the symmetry claims for the truncated parity polytopes,
together with the extension that keeps each $z_i$-variable in
place. For proving the main property, assume that $x_1,\ldots,x_n$ are
all in $\{0,1\}$.  The first part~\eqref{eqn:firstpart}
of~$\mathrm{EX}_{n,t}(x_1,\ldots,x_n, y)$ has the feature expressed in
property~\eqref{prop:truncatedpp}. Let then~$z_1,\ldots,z_{|n|}$
satisfy the conclusion in that property. Thus,~$z_1,\ldots,z_{|n|}$
are the flips of the bits of the binary representation of
$\sum_{i=1}^n x_i$.  In particular they are all in $\{0,1\}$. The
second part~\eqref{eqn:secondpart}
of~$\mathrm{EX}_{n,t}(x_1,\ldots,x_n, y)$ has the feature that if all
$z_1,\ldots,z_{|n|}$ are in~$\{0,1\}$, then there is a unique $y$ in
$\reals$ that makes a solution, and this $y$ is precisely the bit that
is $1$ if and only if all~$z_k$ with $k \in K_0$ are one, and all
$z_k$ with $k \in K_1$ are zero. In other words, $y = 1$ if and only
if $z_1,\ldots,z_{|n|}$ are the flips of the bits of the binary
representation of $t$, and hence if and only if $\sum_{k = 1}^n x_k =
t$.
\end{proof}

\subsection{The construction}

Let us recall how AND, OR and NOT gates are represented by LPs. The
LPs for these three types of gates do not require auxiliary variables,
and their size is linear in the fan-in. Define:
\begin{equation*}
\begin{array}{lllll}
\underline{\mathrm{AND}(x_1,\ldots,x_n,y)} & & \underline{\mathrm{OR}(x_1,\ldots,x_n,y)} & & 
\underline{\mathrm{NOT}(x,y)} \\
y \geq \sum_{i=1}^n x_i - n + 1 & & 1-y \geq \sum_{i=1}^n (1-x_i) - n + 1 & & y = 1-x \\
y \leq x_i & & 1-y \leq 1-x_i & & 0 \leq x \leq 1 \\
0 \leq x_i \leq 1 & & 0 \leq x_i \leq 1 & & 0 \leq y \leq 1. \\
0 \leq y \leq 1. & & 0 \leq y \leq 1. & &
\end{array}
\end{equation*}


\noindent The main properties of these LPs are summarized in the following:

\begin{lemma} \label{lem:andornotgates}
The linear programs $\mathrm{AND}(x_1,\ldots,x_n,y)$,
$\mathrm{OR}(x_1,\ldots,x_n,y)$, and $\mathrm{NOT}(x_1,y)$ have size
linear in $n$, are symmetric with respect to the group of
permutations of its variables that fix $y$, and have the following
property: If $x_1,\ldots,x_n$ are all in $\{0,1\}$, with $n = 1$
for~$\mathrm{NOT}$, then there is a unique $y \in \reals$ that makes
$(x_1,\ldots,x_n,y)$ feasible, and this $y$ is the unique output bit
of the corresponding gate evaluated on inputs $x_1,\ldots,x_n$.
\end{lemma}

\begin{proof} For NOT this is totally obvious. 
For AND it is easy to check, and for OR it follows from the corresponding properties 
of AND and NOT.
\end{proof}

We define the conversion from circuits to LPs. Let $U$ be a set and
let $C$ be a circuit with inputs~$(x_u)_{u \in U}$. Let $C'$ be the
circuit that results from replacing each $k$-threshold gate with
inputs $y_1,\ldots,y_m$ by $\bigvee_{t=k}^{m}
\mathrm{EX}_{m,t}(y_1,\ldots,y_m)$, where
$\mathrm{EX}_{m,t}(y_1,\ldots,y_m)$ denotes an exact counting gate
with inputs $y_1,\ldots,y_m$ which outputs $1$ if and only if the
exact number of $1$'s in the input is $t$.  The resulting circuit is
equivalent to $C$ and its number of gates is polynomial in that of
$C$. We define the LP; we call it $\mathrm{LP}(C)$.

For each gate $i$ in $C'$,
let $y_i$ be a variable constrained by the inequalities
\begin{equation}
0 \leq y_i \leq 1.
\end{equation} 
For each gate $o$ in $C'$ add the constraints and the auxiliary
variables, when necessary, that express their computation:
\begin{equation}
\begin{array}{lll}
y_o = x_u & & \text{ if $o$ is an input gate labelled by $x_u$}, \\ 
\mathrm{NOT}(y_i, y_o) & & \text{ if $o$ is a NOT gate with input $i$}, \\
\mathrm{AND}(y_{i_1},\ldots,y_{i_m}, y_o) & &
\text{ if $o$ is an AND gate with inputs $i_1,\ldots,i_m$},
\\ 
\mathrm{OR}(y_{i_1},\ldots,y_{i_m}, y_o) & &
\text{ if $o$ is an OR gate with inputs $i_1,\ldots,i_m$},
\\ 
\mathrm{EX}_{m,t}(y_{i_1},\ldots,y_{i_m}, y_o) & & 
\text{ if $o$ is an EX$_{m,t}$ gate with inputs $i_1,\ldots,i_m$}, \\
y_o = 1 & & \text{ if $o$ is the output gate of the circuit}.
\end{array}
\end{equation}
By Lemmas~\ref{lem:exgates} and~\ref{lem:andornotgates}, all six cases
have size polynomial in the number of inputs, hence the total size is
polynomial in the size of $C'$. In case $C$ does not have TH gates,
the step for replacing them by EX gates is not done, and all gates are
AND, OR, NOT, so the size of the LP stays linear in the size of
$C$. 

\begin{lemma} \label{lem:abstract}
If $U$ is a $G$-set and $C$ is $G$-symmetric, then $\mathrm{LP}(C)$ is
$G$-symmetric and recognizes the same subset of $\{0,1\}^U$ as $C$.
\end{lemma}

\begin{proof}
The claim that $\mathrm{LP}(C)$ recognizes the same subset of
$\{0,1\}^U$ as $C$ follows from Lemmas~\ref{lem:exgates}
and~\ref{lem:andornotgates}. We prove the symmetry. First note that the
intermediate circuit $C'$ is also $G$-symmetric. Now we show that the
LP is also $G$-symmetric. Fix some $\pi \in G$. Let~$\sigma$
be a permutation of the gates of $C'$ so that the pair $(\pi,\sigma)$
leaves $C'$ in place. In particular, for each gate $o$ of $C'$ with
inputs $i_1,\ldots,i_m$, if $p = \sigma(o)$, then $p$ is the same type
of gate as $o$, has the same fan-in $m$, and if $o$ is an input gate
fed by $x_u$, then $p$ is an input gate fed by~$x_{\pi(u)}$. Moreover,
if~$j_1,\ldots,j_{m}$ are the inputs of gate $p$, then there is a
permutation~$\tau_o \in \Sym_m$ so that $\sigma(i_k) = j_{\tau_o(k)}$
for every $k \in [m]$. If we think of $\sigma$ as mapping the output
variable $y_o$ of the linear program $P_o =
P(y_{i_1},\ldots,y_{i_m},y_o)$ for gate $o$ to the output variable
$y_p$ of the linear program $P_p = P(y_{j_1},\ldots,y_{j_m},y_p)$ for
gate $p$, then, by the above, this map takes the input variables of
$P_o$ to the input variables of $P_p$. We want to show that
this~$\sigma$ can be extended to also map the auxiliary variables of
$P_o$ to the auxiliary variables of $P_p$ in such a way that the
resulting extension of $\pi$ is an automorphism of the linear program
$P$ for~$C'$. We define this extension automorphism gate by gate.

We start with the internal gates of $C'$. By the symmetry claims in
Lemmas~\ref{lem:exgates} and~\ref{lem:andornotgates}, the permutation
that agrees with $\tau_o$ on the input variables~$y_{j_1},\ldots,y_{j_m}$
of $P_p$ and that fixes the output variable $y_{p}$, extends to an
automorphism~$\rho_{o}$ of $P_p$. With the automorphisms $\rho_o$ in
hand, we are ready to define the automorphism of $P$ that extends
$\pi$: for each gate $o$, map the variable $y_o$ to $y_{\sigma(o)}$,
and map each auxiliary variable of~$P_o$ to the image under $\rho_o$
of the corresponding auxiliary variable of $P_p$, where $p =
\sigma(o)$. For all internal gates, this has the required properties
by construction. For the gates $o$ that are fed by a variable~$x_u$,
the gate~$\sigma(o)$ must be fed by the variable $x_{\pi(u)}$, and
therefore $y_o = x_u$ gets mapped to~$y_{\sigma(o)} = x_{\pi(u)}$, as
required. For the output gate $o$ of $C'$ we have $\sigma(o) = o$, and
the constraint $y_o = 1$ gets mapped to itself.
\end{proof}

\begin{proof}[Proof of Lemma~\ref{lem:circuitstolps}]
This is a consequence of Lemma~\ref{lem:abstract}: for the $n$-th
LP we let $U = L(n)$, 
let $G = \Sym_n$ with the natural action on $L(n)$, 
and use $\mathrm{LP}(C_n)$, where $C_n$ is
the~$n$-th circuit.
\end{proof}

%% file: section-4-from-lps-to-fpc.tex
\section{From LPs to Logic}\label{sec:fromLPstocircuits}

We say that a
function $s(n)$ is at most weakly exponential if there exists a
positive real $\epsilon$ such that $s(n) \leq 2^{n^{1-\epsilon}}$ for
every sufficiently large $n$.
In this section we prove the second half of the equivalence
which takes families of symmetric linear programs to families of
formulas of counting logic. That is:

\begin{lemma}\label{lem:secondhalf}
If $\mathscr{C}$ is a class of $L$-structures that is recognized by a family of
$L$-symmetric polytope lifts of size $s(n)$, then $\mathscr{C}$
is recognized
by a family of $\mathrm{C}^{k(n)}$ formulas, where $k(n) =
O(\log(s(n))/(\log(n)-\log\log(s(n))))$.
Moreover, if $s(n)$ is at most weakly exponential,
then the formulas have size $s(n)^{O(1)}$.
\end{lemma}

The key technical tool 
is the notion of a \emph{bounded support}.
The existence of bounded supports implies that
a property of $L$-structures 
recognised by a family of $L$-symmetric LP lifts
is recognised by a family of \emph{manageable} such LP lifts. 
Let us illustrate the notion of a manageable linear program
by an example.

Consider a graph-symmetric polytope $P$
over $\reals^{[2]^2} \times \reals^{2}$ given by the following
system of linear constraints:
\begin{align*}
& x_{11}    \leq 1 \\
& x_{22}  \leq 1 \\
& x_{12} - y_{12} \leq 0 \\
& x_{21} - y_{21} \leq 0 \\
& y_{12} + y_{21} \leq 1
\end{align*}
Two properties of $P$ are central to the notion
of a manageable polytope. 
Firstly, each of the auxiliary variables $y_{12}, y_{21}$
is essentially a tuple
of integers from $[2]$ of a bounded length.
Secondly, $P$ is graph-symmetric
with respect to
the natural action of the group $\Sym_2$ on
the set $\{ y_{12}, y_{21} \}$. 
Indeed, for any permutation $\pi \in \Sym_2$ the
system of constraints:
\begin{align*}
& x_{\pi(1)\pi(1)}    \leq 1 \\
& x_{\pi(2)\pi(2)}   \leq 1 \\
& x_{\pi(1)\pi(2)} - y_{\pi(1)\pi(2)} \leq 0 \\
& x_{\pi(2)\pi(1)} - y_{\pi(2)\pi(1)} \leq 0 \\
& y_{\pi(1)\pi(2)} + y_{\pi(2)\pi(1)} \leq 1
\end{align*}
defines the same polytope.
A formal definition of a manageable polytope is
deferred to Subsection~\ref{sec:manageable}.

Given an $L$-structure over an $n$-element domain,
a manageable polytope lift
over $\reals^{L(n)} \times \reals^W$
can be turned into a linear program whose variables
are indexed by tuples of elements of the structure of bounded
length. Moreover, the symmetry condition guarantees
that this can be done
without referring to any concrete bijection between
$[n]$ and the domain of the structure.
Thus, as we show, it can be performed
by means of logical interpretations.

In what follows, from a family $(P_n)_{n \in \naturals}$
of $L$-symmetric LP lifts
we obtain a family $(\bar{P}_n)_{n \in \naturals}$ of
manageable LP lifts recognising the same family of structures.
Lemma~\ref{lem:rigidity} below implies
that from each $P_n$
one can construct a polytope lift $\widehat{P}_n$
which recognises the same subset of $\{0,1\}^{L(n)}$
but comes with an action of the group $\Sym_n$ witnessing its
symmetry. Further, in Subsection~\ref{sec:supports}
we show that the action of $\Sym_n$ on each of the
constraints and
auxiliary variables of $\widehat{P}_n$ depends on a subset of $[n]$ of
bounded size called its support.
In the second part of Subsection~\ref{sec:supports}
we analyse properties
of sets whose elements have bounded supports in order to show
that they are essentially sets of tuples of integers from $[n]$.
This implies, in Subsection~\ref{sec:manageable},
that each $\widehat{P}_n$ after a small modification becomes
a manageable LP lift $\bar{P}_n$. 
Finally, in Subsection~\ref{sec:translation}
based on $\bar{P}_n$ we construct a $\FOC$-interpretation
that given an $L$-structure $\mathbb{A}$
over an $n$-element domain 
outputs a linear program
which has a solution if and only if $\mathbb{A}$ belongs to the class of interest (for a definition of an interpretation
see e.g.~\cite{hodges_1993}).
Since solving linear programs is expressible in FPC~\cite{Anderson:2015},
we are able to conclude the proof.


\subsection{Rigid polytopes}\label{sec:rigid}

In this subsection we consider
general $G$-symmetric LPs,
i.e., not necessarily $L$-symmetric.

Let $U$ be a $G$-set and let
$P \subseteq \reals^U \times \reals^W$ be a 
$G$-symmetric polytope given
by a sequence of linear 
constraints $(\gamma_v)_{v \in V}$
where each 
$\gamma_v$ is of the form
$ \avector^\transpose\xvector + \bvector^\transpose\yvector  \leq c$,
with $\xvector = (x_u)_{u \in U}$ and $\yvector
= (y_w)_{w \in W}$. 
We say that the polytope $P$ is \emph{rigid} if for every $\pi \in G$ 
there exists a unique element of $\Sym_W$,
let us denote it by $\sigma_{\pi}$,
such that
$P^{(\pi, \sigma_{\pi})} = P$.

Assume that $P$ is rigid. It is easy see that
the mapping from
$G$ to $\Sym_W$ given by $\pi \mapsto \sigma_\pi$
is a group homomorphism. 
Hence, there is
a natural action of the
group $G$ on the set of auxiliary variables $\{y_w\}_{w \in W}$
such that for any $\pi \in G$ and $w \in W$
applying $\pi$ to $y_w$ gives $y_{\sigma_\pi(w)}$,
and the polytope $P$ is $G$-symmetric
with respect to this action.
Moreover, this induces
an action of the group $G$ on the set of
linear constraints
$\{\gamma_v\}_{v \in V}$ in the obvious way:
for any $\pi \in G$ and any $v \in V$
applying $\pi$
to $\gamma_v$
of the form $\avector^\transpose\xvector + \bvector^\transpose\yvector \leq c$ gives
$\avector^\transpose\xvector^{\pi} +
\bvector^\transpose\yvector^{\sigma_\pi}  \leq c$, and the symmetry of
$P$ guarantees that this is also a constraint.
For rigid $G$-symmetric polytopes, we write
$\yvector^{\pi}$ to mean $\yvector^{\sigma_\pi}$, we 
use $\gamma_v^{\pi}$ to denote
$\avector^\transpose\xvector^{\pi} + \bvector^\transpose\yvector^{\sigma_\pi}  \leq c$, and $P^\pi$ to denote $P^{(\pi,\sigma_\pi)}$.

Suppose that a subset $A$ of $\{0,1 \}^U$
is recognised by a $G$-symmetric polytope lift~$P$.
We show that there exists 
a rigid $G$-symmetric polytope lift $\widehat{P}$
of size polynomial in the size of $P$ recognising $A$.
More precisely, the number of auxiliary variables
and the number of constraints of $\widehat{P}$ are at most, respectively, the number of
auxiliary variables and the number of constraints of $P$,
and the bit-size of the coefficients which appear in
the linear constraints defining $\widehat{P}$ is polynomial in the bit-size of the coefficients which appear in
the linear constraints defining~$P$.
 

The construction of $\widehat{P}$ goes as follows.
For the subgroup of $\Sym_W$
consisting of all permutations $\sigma$ such that $P^{(\id, \sigma)} = P$,
consider the orbits of the set of auxiliary variables
$\{y_w\}_{w \in W}$ under the action of this subgroup.
By identifying the variables in each of those orbits we obtain
a new $G$-symmetric polytope lift recognising $A$ with potentially
smaller number of auxiliary variables.
This procedure
needs to be iterated
until the obtained polytope is rigid. 
Let us provide more details.

We need a few pieces of notation.
For every $\pi \in G$, by $\ext(\pi)$ we denote the set of all
$\sigma \in \Sym_W$ satisfying $P^{(\pi,\sigma)} = P$ and 
by~$H$ we denote the union of
$\ext(\pi)$ over all the elements~$\pi$ of~$G$, i.e.,
$H = \{ \sigma \in \Sym_W  :   \text{ there exists } \pi \in G \text{ such that } P^{(\pi, \sigma)} = P \}$.
It is easy to see that both $\ext(\id)$ and $H$ are subgroups of $\Sym_W$.
Moreover, for any $\pi \in G$ the set $\ext(\pi)$
is a coset of $\ext(\id)$ in $H$.
Now let $K = \sum_{\pi \in G} \{\pi \} \times \ext(\pi) $.
Observe that $K$ is a subgroup of the direct product $G \times H$ and
consider the projection homomorphism $f_{G} : K \rightarrow G$
given by $f_G(\pi, \sigma) = \pi$. The kernel of this homomorphism is
the group $\{ \id \} \times \ext(\id)$. Let us denote
it by $J$.
Since the homomorphism $f_G$ is surjective, 
the quotient $K / J$ is isomorphic to $G$.

Notice that the polytope $P$ is rigid if, and only if, 
the group $\ext(\id)$ is
trivial. Assume therefore that this is
not the case and let $\mathcal{O}$ denote the set
of orbits of the set of auxiliary variables $\{y_w\}_{w \in W}$
under the action of $\ext(\id)$.
Recall that $\mathcal{O}$ forms a partition of $\{y_w\}_{w \in W}$.
For every orbit $O \in \mathcal{O}$, 
we introduce a new variable $y_{O}$ and
for any polytope $R$,
let $\widehat{R}$ denote the polytope obtained from $R$
by substituting, for each $O \in \mathcal{O}$, every variable in
$O$ (if present) by $y_{O}$.
We aim to show that $\widehat{P}$ is a $G$-symmetric
polytope recognizing $A$.



The projection homomorphism $f_{H} : K \rightarrow H$
given by $f_H(\pi, \sigma) = \sigma$
can be seen as a homomorphism from $K$ to $\Sym_W$ and
hence defines an action of the group $K$ on 
the set of variables $\{y_w\}_{w \in W}$.
Since $J$ is a normal subgroup of $K$,
the quotient group $K / J$ acts on the set of orbits of 
$\{y_w\}_{w \in W}$ under the (induced) action of $J$:
for any $\pi \in G$, $\sigma \in \ext(\pi)$ and $w \in W$,
applying $(\pi,\sigma)J$ to the orbit of $y_w$ maps it
to the orbit of $y_{\sigma(w)}$.  
Since the groups $J$ and $\ext(\id)$ as well as
the groups $K /J$ and $G$ are isomorphic,
this gives us an action of $G$ 
on $\mathcal{O}$ which, as we argue below, witnesses the
fact that the polytope $\widehat{P}$ is $G$-symmetric.
Unfolding the abstract definition of this group action,
there is a homomorphism $h$ from the group $G$ to
the symmetric group $\Sym_{\mathcal{O}}$ such that
for every $\pi \in G$, $h(\pi)$ is the permutation of $\mathcal{O}$,
which for any $w \in W$,  
maps the orbit of the variable $y_w$
to the orbit of the variable $y_{\sigma(w)}$,
where $\sigma$ is any permutation from $\ext(\pi)$.
Finally, let us observe that,
for any $\pi \in G$ and $\sigma \in \ext(\pi)$, it holds that 
$\widehat{P}^{(\pi,h(\pi))} = \widehat{P^{(\pi, \sigma)}} = \widehat{P}$,
which means that indeed $\widehat{P}$ is $G$-symmetric.
It remains to show that $\widehat{P}$ recognises $A$.

Let $\xvector \in A$ and take some $\yvector \in \reals^W$
such that $(\xvector, \yvector) \in P$. 
Note that for every $\sigma \in \ext(\id)$ it holds that
$(\xvector, \yvector^\sigma) \in P$. Hence, we have
$(\xvector, \yvector') \in P$ where
$\yvector' = \sum_{\sigma \in \ext(\id)} \yvector^\sigma / |\ext(\id)|$.
Now for every $O \in \mathcal{O}$ let $p_O$
be the sum of the values taken
by the variables from $O$ in the solution $(\xvector, \yvector)$.
In the solution $(\xvector, \yvector')$ every variable
$y$ in $\mathcal{O}$ takes value 
$p_O |\ext(\id)_{y}| / |\ext(\id)| = p_O / |\mathcal{O}|$, where $\ext(\id)_{y}$
denotes the stabiliser of $y$ in $\ext(\id)$.
Hence, by assigning for every $O \in \mathcal{O}$
the value $p_O$ to
the variable $y_O$
we obtain a point $(\xvector, \widehat{\yvector})$ in $\widehat{P}$. 
This implies that $A$ is contained in the subset of $\{0,1\}^U$ recognised by
$\widehat{P}$.
The other inclusion is clear.

We are now ready to state the main conclusion of this subsection.
\begin{lemma}\label{lem:rigidity}
  For every $G$-symmetric polytope $P$ of size $s$, there is a rigid $G$-symmetric
  polytope $Q$ which recognises the same set and the size of $Q$ is
  not more than $s \log(s)$.
\end{lemma}
\begin{proof}
  First note that for any $G$-symmetric polytope $R$, 
  if $R$ is not rigid,
  then $\widehat{R}$ has strictly fewer variables than $R$.  Thus,
  starting at $P$, if we iterate the process, we must, in a finite
  number of steps reach a rigid polytope $Q$.  For the size bound,
  note that the number of variables and constraints in $Q$ is at most
  the corresponding number in $P$.  Moreover, each coefficient in $Q$
  is the sum of distinct coefficients in $P$, of which there are at
  most $s$.  Thus, if each
  coefficient in $P$ can be expressed with $b$ bits, each coefficient
  in $Q$ requires at most $b \log(s)$ bits and the bound follows.
\end{proof}

\subsection{Bounded supports}\label{sec:supports}

For a $\Sym_n$-set $Y$,
a subset $S$ of $[n]$ is said to be a \emph{support} of
an element $y \in Y$
if for every $\pi \in \Sym_n$ that fixes $S$ pointwise,
it holds that $\pi \cdot y = y$. Intuitively, this means
that the action of the group $\Sym_n$ on $y$
depends only on the set $S$. We get even supports by replacing the symmetric group $\Sym_n$ by the alternating group $\Alt_n$.
A subset $S$ of $[n]$ is said to be an \emph{even support}
of $y \in Y$ if for every $\pi \in \Alt_n$ that
fixes $S$ pointwise, we have $\pi \cdot y = y$.
Clearly, any support of $y$ is also an even support of $y$.
Note also that every element of a $\Sym_n$-set
is supported by the whole set $[n]$.

For a non-negative integer $k$,
an (even) support $S$ is called $k$-\emph{bounded} if
$|S| \leq k$.
A $\Sym_n$-set $Y$ is called $k$-\emph{supported}
if each element of $Y$ has a $k$-bounded support.
An $L$-symmetric polytope $P$ is called $k$-\emph{supported}
if the set of auxiliary variables
and the set of constraints of $P$ are $k$-supported.
We now show the following:

\begin{lemma} \label{lem:supports}
  There exists a positive integer $n_0$ such that for
  any positive integers $s$ and $n$
  satisfying $s \geq n \geq n_0$, the
  following holds: If $P$ is a rigid $L$-symmetric LP lift of size $s$
  for structures with $n$ elements, then $P$ is $k$-supported,
  where~$k =
  O(\log(s)/(\log(n)-\log\log(s)))$.
  Moreover, if~$s \leq 2^{n/3}$, then the size of
  $P$ is at most $n^k$.
\end{lemma}

\begin{proof}
For the sake of simplicity we give the proof
for the case when $L$ consists of a single binary symbol, that
is for the case of graphs. The general case is completely
analogous.


Consider a rigid graph-symmetric polytope lift
 $P \subseteq \reals^{[n]^2}\times \reals^W$ of size $s$
which recognises some property of graphs with $n$ vertices,
and let $P$ be given
by a sequence of linear 
constraints $(\gamma_v)_{v \in V}$ where each 
$\gamma_v$ is of the form
$ \avector^\transpose\xvector + \bvector^\transpose\yvector \leq c$,
with $\xvector = (x_{ij})_{i,j \in [n]}$ and $\yvector
= (y_w)_{w \in W}$.
Since $P$ is rigid, it comes with an action of the group
$\Sym_n$ on the set of auxiliary variables $\{y_w\}_{w \in W}$
such that for every $\pi \in \Sym_n$ we have $P^\pi = P$,
and with an induced action of the group $\Sym_n$
on the set of constraints $\{\gamma_v\}_{v \in V}$.

If the size $s$ of the linear program $P$ is
greater than $2^{n/3}$, we can take $k = n$.
Observe that in case $s > 2^{n/3}$ we have
  \begin{equation}
    \frac{\log(s(n))}{\log(n)-\log\log(s(n))} \geq
    \frac{n/d_0}{\log(n)-\log(n)+\log(3)} = \frac{n}{3 \log(3)}. \label{eqn:stuff1}
  \end{equation}
  It follows that $k = O(\log(s)/(\log(n)-\log\log(s)))$.
  Since any element of a $\Sym_n$-set is supported by $[n]$,
for this choice of $k$,
every auxiliary variable and every constraint of $P$
has a $k$-bounded support.


In the case $s \leq 2^{n/3}$ the argument is more involved.
First we obtain bounded even supports.
Take $t =
  \log(s)/(\log(n)-\log\log(s))$ and $k = \lceil{t}\rceil$.  Observe
  that the denominator in the definition of $t$ is non-zero, since by
  assumption we have $s \leq 2^{n/3} < 2^n$.
Also, we
  have $0 < t \leq k \leq n/3\log(3) < n/4 < n/e$, to be used later in the proof. 
  Let us
  start by noting that
  \begin{equation}\label{eqn:tbound}
  t\log\left({\frac{n}{t}}\right) = \log(s)\frac{\log(n)-\log\log(s)+\log(\log(n)-\log\log(s))}{\log(n)-\log\log(s)} > \log(s). 
  \end{equation}
  The inequality follows from the fact that the big fraction in the
  middle is strictly bigger than~$1$ since $s \leq 2^{n/3}$.

For any $S \subseteq [n]$, let
$\Alt_{(S)}$ denote the group of all even permutations of $[n]$
that fix the set $S$ pointwise. 
We use the following fact which guarantees that subgroups of the
symmetric group with small index 
contain as subgroups large
alternating groups.
\begin{lemma}[see Theorem~5.2B in~\cite{Dixon-Mortimer}]\label{lem:even-support}
If $n > 8$ and $1 \leq k \leq n/4$, and $G$ is a subgroup of $\Sym_n$
with $[\Sym_n : G]  < {n \choose k}$, then there
is a set $S \subseteq [n]$ with $|S| < k$ such that
$\Alt_{(S)} \leq G$.
\end{lemma}

For $w \in W$, let $\Stab_w$ denote the stabilizer of $y_w$
in $\Sym_n$, i.e., the subgroup of $\Sym_n$ defined by
$\Stab_w = \{ \pi \in \Sym_n : \pi \cdot y_w = y_w \}$.
Since $[\Sym_n : \Stab_w]$ is the size of the orbit of $y_w$ under
the action of $\Sym_n$ and the total number of
auxiliary variables is bounded by the size of $P$, we have
  \begin{equation}
    [\Sym_n : \Stab_i] \leq s < \left({\frac{n}{t}}\right)^t
    \leq \left({\frac{n}{k}}\right)^k \leq {n \choose k}
  \end{equation}
  with the second following from~\eqref{eqn:tbound}, and the third from
  $0 < t \leq k < n/4 < n/e$ and the fact that~$f(x) = (n/x)^x$ is an increasing
  function of $x$ in the interval
  $(0,n/e)$. Lemma~\ref{lem:even-support} implies that,
  if $n$ is large enough, there exist
  $S \subseteq [n]$ with $|S| < k$ and $\Alt_{(S)} \leq
  \Stab_w$. This is a $k$-bounded even support of $y_w$
in the way we defined.
An entirely analogous argument yields a $k$-bounded even support
for each constraint in $\{\gamma_v\}_{v \in V}$.

In order to obtain supports in place of even supports
we need to introduce a way of looking at automorphism
groups of symmetric polytopes as automorphism groups of
graphs. 
In the following, whenever we talk about coloured vertices,
by a colour we
mean a graph gadget that forces the vertices of
a colour to be mapped to vertices of the same colour
by every possible automorphism.

Let $C$ be the set
of all different numerical coefficients which appear
in the linear constraints representing $P$.
A \emph{graph representation} of the polytope $P$
is a graph $\mathbb{P}$ with:
\begin{enumerate}\itemsep=0pt
\item five disjoint sets of vertices: $[n]$,
$\{x_{ij}\}_{i,j \in [n]}$,
$\{y_w\}_{w \in W}$, $\{\gamma_v\}_{v \in V}$ and $C$, 
\item the vertices in $[n]$, $\{x_{ij}\}_{i,j \in [n]}$,
$\{y_w\}_{w \in W}$ and $\{\gamma_v\}_{v \in V}$ 
coloured with four different colours depending on the set they belong to,
\item each vertex in $C$ coloured with a unique colour, which
does not appear anywhere else in the graph,
\item for any $i,j \in [n]$, an edge from $i$ to $x_{ij}$ and from $j$ to $x_{ij}$,
\item for any $v \in V$ and any variable which appears in the
constraint $\gamma_v$, a pair of edges: from the variable to
its coefficient, and from this coefficient to $\gamma_v$,
\item for any $v \in V$,
an edge from $\gamma_v$ to the vertex in $C$ corresponding to 
its constant term.
\end{enumerate}
Observe that the automorphism group of the graph $\mathbb{P}$ is
isomorphic
to the automorphism group of $P$. Also, the number of vertices in
$\mathbb{P}$ can be bounded by $O(s^2)$.  This is because the number
of vertices introduced in item (1) is at most $s$ and each gadget
introduced in (2) and (3) can be chosen to have at most $2s$ vertices
and there are at most $s$ of them.

For any $S \subseteq [n]$, let
$\Sym_{(S)}$ denote the group of all permutations of $[n]$
that fix the set $S$ pointwise. 
Take some $w \in W$ and let $S$ be a
$k$-bounded even support of $y_w$. Since $\Alt_{(S)} \leq \Stab_w$,
we have $\Alt_{(S)} \leq \Stab_w \cap \Sym_{(S)}  \leq \Sym_{(S)}$.
Hence, $\Stab_w \cap \Sym_{(S)} = \Alt _{(S)}$ or $\Stab_w \cap \Sym_{(S)} = \Sym _{(S)}$. We argue it is the latter case that holds using the following
theorem which states that a graph whose automorphism group
is the alternating group on an $n$-element set
must be of size exponential in $n$.

\begin{lemma}[Theorem~A in~\cite{Liebeck}]\label{lem:alt}
If $n > 22$, then the number of vertices of any graph whose full
automorphism group is isomorphic to $\Alt_n$ is at least 
$1/2 {n \choose \lfloor n/2 \rfloor } \sim  { 2^n / \sqrt{2 \pi n} }$ .
\end{lemma}

Assume that $\Stab_w \cap \Sym_{(S)} = \Alt _{(S)}$. We use
Lemma~\ref{lem:alt} to arrive at a contradiction. Consider the graph
representation $\mathbb{P}$ of the polytope $P$ modified in the
following way: the vertices
in $S \subseteq [n]$ and the vertex $y_w$ are coloured
each with a different colour which did not appear in the graph before.
Observe that the automorphism group of the graph $\mathbb{P}_w$
obtained this way is isomorphic to $\Stab_w \cap \Sym_{(S)}$, and
therefore isomorphic to $\Alt _{(S)}$, which in turn is isomorphic to
the alternating group on the set $[n - |S|]$.  And, once again, the
number of vertices of $\mathbb{P}_w$ is $O(s^2)$.  Thus, if $n$ is large
  enough, we have
  \begin{equation}
    s^2 < 2^{2n/3} < \frac{1}{2}{n \choose \lfloor{n/2}\rfloor}.
  \end{equation}
Hence, by Lemma~\ref{lem:alt}, we obtain
the desired contradiction.

From $\Stab_w \cap \Sym_{(S)} = \Sym _{(S)}$ it follows that
$\Sym_{(S)} \leq \Stab_w$, which means that $S$ is a $k$-bounded
support of $y_w$ in the way we defined. An analogous argument
yields a $k$-bounded support for each linear constraint in
$\{\gamma_v\}_{v \in V}$.
Note also that that $s < \binom{n}{k} \leq n^k$.
In particular,
$P$ has at most
  $n^k$ auxiliary variables, at most $n^k$ constraints,
  and all its coefficients and constant terms can be
  written down using at most $n^k$ bits.
\end{proof}

%
%
%
%

We now show that it is possible to (non-uniquely) represent
the auxiliary
variables and constraints of $k$-supported polytopes
by tuples of integers from $[n]$
of length $k$ in a way that is consistent with the
group action. In order for the representation to be uniform
across all $n$,
we extend the definition of the set $[n]^{(k)}$ to the case when
$k > n$.
For positive integers $n$ and $k$ such that $k > n$,
the set $[n]^{(k)}$ consists of 
$k$-tuples of elements of $[n]$ with the first $n$
components pairwise distinct and the last $k-n$
components equal to the $n$-th component.
In particular, if $k > n$, then
every tuple in $[n]^{(k)}$ contains all the elements of $[n]$.

For any positive integer $n$
and any non-negative integer $k$,
we consider the set $[n]^{(k)}$ 
as a $\Sym_n$-set with the natural action of the group $\Sym_n$.
Note that this $\Sym_n$-set has 
one orbit. Indeed, if we take any $k$-tuples
$\svector_1$ and $\svector_2$ from $[n]^{(k)}$,
then there exists a permutation $\pi$ such that
$\pi \cdot \svector_1 = \svector_2$.
This is because the equality types of any
two elements of $[n]^{(k)}$ are the same.

\begin{lemma}\label{lem:representation}
Let $Y$ be a single-orbit $k$-supported $\Sym_n$-set.
There is a surjective homomorphism from
$[n]^{(k)}$ to $Y$. 
\end{lemma}

\begin{proof}
Take any $y \in Y$ and let $S$ be a $k$-bounded support of $y$.
Since a superset of a support
is a support itself,
without loss of generality we can assume
that $|S| = \min\{k,n\}$. Now, pick a tuple $\svector \in S^{(k)}$.
If $k \leq n$, then $|S| = k$ and every element of $S$ appears exactly
once in the tuple $\svector$, otherwise the tuple
$\svector$ contains all the elements of $[n]$.

We define a homomorphism $f$ from the $\Sym_n$-set
$[n]^{(k)}$ to the $\Sym_n$-set $Y$
which for any $\pi \in \Sym_n$ maps $\pi \cdot \svector$ to $\pi \cdot y$.
The only thing that needs to be verified
is whether the function $f$ is well defined.
To this end, suppose that for some permutations $\pi_1, \pi_2 \in \Sym_n$
it holds that $\pi_1 \cdot \svector = \pi_2 \cdot \svector$. Then 
$\pi_2^{-1}\pi_1 \cdot \svector = \svector$,
that is, the permutation $\pi_2^{-1}\pi_1$ fixes the
support $S$ of~$y$
pointwise. Therefore, $\pi_2^{-1}\pi_1 \cdot y = y$
which implies that $\pi_1 \cdot y = \pi_2 \cdot y$.
Since $Y$ has one orbit,
the homomorphism $f$
is surjective. 
\end{proof}

Once a surjective homomorphism $f$ from a $\Sym_n$-set $[n]^{(k)}$
to a $\Sym_n$-set
$Y$ is fixed, the family
$\{f^{-1}(y)\}_{y \in Y}$ forms a partition of $[n]^{(k)}$.
Hence, for any $y \in Y$, each tuple $(i_{1},\ldots, i_{k})$ from
$f^{-1}(y)$ uniquely identifies $y$, and is called an \emph{identifier} of $y$ (with respect to
the homomorphism $f$).
In most cases each element of $Y$ has several identifiers.

Let us illustrate Lemma~\ref{lem:representation} 
by an example. For $n \geq 2$,
consider the set $Y$ of two-element subsets of
$[n]$ with the natural action of the group $\Sym_n$, i.e.,
for any $\pi \in \Sym_n$ and any distinct $i, j \in [n]$, 
we have $\pi \cdot \{ i, j \} = \{ \pi(i), \pi(j) \}$. 
This single-orbit $\Sym_n$-set is $k$-supported,
for any $k \geq 2$.
Let us take $k = 2$. The homomorphism from Lemma~\ref{lem:representation}
is then unique and given by $(i,j) \mapsto \{i,j\}$. Note that
the inverse image of any $\{i,j\} \in Y$ has two elements,
and hence,
each element $\{i,j\}$ of $Y$ has two identifiers: $(i,j)$ and $(j,i)$.
For $k = 3$, applying Lemma~\ref{lem:representation} yields
several different homomorphisms. If $n \geq 3$, one of them
is given by $(i,j,*) \mapsto \{i,j\}$, where by $*$ we
mean any element of $[n]$ distinct from both $i$ and $j$.
In this case, the inverse image of any $\{i,j\} \in Y$ has $2n-4$
elements. For $n=2$ and $k=3$, the homomorphism
is again unique and given by $(i,j,j) \mapsto \{i,j\}$
yielding two identifiers, for every element of $Y$.

To give one more example, consider a single-orbit set
which consists of a single element $y$ with the trivial action of the
symmetric group $\Sym_n$. This set 
is $k$-supported, for any non-negative integer $k$.
For any $k \geq 0$, the homomorphism
from Lemma~\ref{lem:representation} mapa each
tuple from $[n]^{(k)}$ to $y$.
In particular, for $k = 0$, we have $\epsilon \mapsto y$,
where $\epsilon$ is the only element of $[n]^{(0)}$, that is,
the empty tuple.

To represent
elements of a $k$-supported $\Sym_n$-sets
with potentially more than one orbit, we need to introduce 
several copies of the set $[n]^{(k)}$, one for each
orbit.

\begin{corollary}\label{cor:many_orbits}
Let $Y$ be a $k$-supported $\Sym_n$-set.
There is a surjective homomorphism from
$Q \times [n]^k$
to $Y$, where the size of $Q$ is equal to
the number of orbits of $Y$.
\end{corollary}

It is clear how to extend the definition of an identifier to the general
case discussed in the corollary above.
Note that if a tuple $(q, i_{1},\ldots, i_{k})$ is an identifier of $y \in Y$,
then
the tuple $(q, \pi(i_{1}),\ldots, \pi(i_{k}))$ is an identifier of
$\pi \cdot y$.

\subsection{Manageable polytopes}\label{sec:manageable}

For a non-negative integer $k$, a polytope $P$ over
$\reals^{L(n)} \times \reals^W$ is called
$k$-\emph{manageable} if:
\begin{enumerate}\itemsep=0pt
\item there are two sets $Q$ and $T$
with a trivial action of the group $\Sym_n$,
\item the set of constraints of $P$ is indexed by
$V = Q \times [n]^k$,
\item the set of auxiliary variables of $P$ is indexed by
$W = T \times [n]^k$,
\item $P$ is $L$-symmetric with respect to
the natural action of $\Sym_n$ on
$W$,
and the induced action of $\Sym_n$ on the set of constraints
is exactly the natural action of~$\Sym_n$ on
$V$.
\end{enumerate}
An example of a manageable polytope is given at the very beginning of
this section. There the set of auxiliary variables has only one orbit
$\{ y_{12}, y_{21} \}$ so introducing $T$ is not necessary.

The key property of $k$-manageable polytopes,
which allows us to use them in the translation 
from families of linear programs to logic, is the following.

\begin{lemma}\label{lem:propertymanageable}
If $P$ is a $k$-manageable polytope with
constraints indexed by
$V = Q \times [n]^{(k)}$ and
auxiliary variables indexed by
$W = T \times [n]^{(k)}$, then
for any $R \in L$, $q \in Q$, $t \in T$, $\mathbf{i}, \mathbf{i}',
 \mathbf{j}, \mathbf{j}' \in [n]^{(k)}$, $\mathbf{k}, \mathbf{k}' \in
[n]^{\mathrm{ar}(R)}$:
\begin{enumerate}\itemsep=0pt
\item 
the constant terms of the linear constraints 
$\gamma_{(q,\mathbf{i})}$ and
$\gamma_{(q,\mathbf{i}')}$ are the same,
\item if the equality types of the tuples
$(\mathbf{j},\mathbf{i})$ and
$(\mathbf{j}',\mathbf{i}')$ are the same, then
the coefficient of the variable $y_{(t,\mathbf{j})}$
in the linear constraint $\gamma_{(q,\mathbf{i})}$
is the same as the coefficient of the variable~$y_{(t,\mathbf{j}')}$
in the linear constraint $\gamma_{(q,\mathbf{i}')}$,
\item if the equality types of the tuples
$(\mathbf{k}, \mathbf{i})$ and
$(\mathbf{k}', \mathbf{i}')$ are the same, then
the coefficient of the variable $x_{(R,\mathbf{k})}$
in the linear constraint $\gamma_{(q,\mathbf{i})}$
is the same as the coefficient of the variable 
$x_{(R,\mathbf{k}')}$
in the linear constraint $\gamma_{(q,\mathbf{i}')}$.
\end{enumerate}
\end{lemma}

\begin{proof}
1. Recall that every tuple in the set $[n]^{(k)}$ 
has the same equality type. Therefore, there exists a permutation
$\pi \in \Sym_n$ which maps the tuple
$\mathbf{i}$ to $\mathbf{i}'$. 
Since $\gamma_{(q,\mathbf{i})}^{\pi} =
\gamma_{(q,\mathbf{i}')}$, the constant terms in the
linear constraints
$\gamma_{(q,\mathbf{i})}$ and
$\gamma_{(q,\mathbf{i}')}$ are the same. 

2. If the equality types of the tuples
$(\mathbf{j},\mathbf{i})$ and
$(\mathbf{j}', \mathbf{i}')$ are the same,
then there exists a permutation $\pi \in \Sym_n$,
which maps $\mathbf{j}$ to $\mathbf{j}'$,
and $\mathbf{i}$ to $\mathbf{i}'$.
Let $a$ be the coefficient of the variable $y_{(t,\mathbf{j})}$
in the linear constraint $\gamma_{(q,\mathbf{i})}$.
By applying the permutation $\pi$ to the constraint
$\gamma_{(q,\mathbf{i})}$, we get the constraint
$\gamma_{(q,\mathbf{i}')}$. Moreover, since 
$\pi \cdot y_{(t,\mathbf{j})} = y_{(t,\mathbf{j}')}$,
the coefficient of the variable $y_{(t,\mathbf{j}')}$ in
$\gamma_{(q,\mathbf{i}')}$ is $a$.
The proof of 3.\ is analogous.
\end{proof}

Now, suppose that a $k$-supported rigid $L$-symmetric LP lift 
$P \subseteq \reals^{L(n)} \times \reals^W$
recognizes some property of $L$-structures, that is, 
a subset $A$ of $\{0,1 \}^{L(n)}$.
We argue that there exists 
a $k$-manageable polytope lift $\bar{P}$
recognising $A$.
Since the polytope $P$ is $k$-supported,
by applying Lemma~\ref{lem:representation} we obtain 
two sets of identifiers: 
$\bar{V} = Q \times [n]^k$ for
the constraints, and
$\bar{W} = T \times [n]^k$
for the auxiliary
variables. 
Let us introduce a new variable of the form
$y_{(t, \mathbf{j})}$, for any identifier 
$(t, \mathbf{j}) \in \bar{W}$.
We obtain a manageable polytope $\bar{P}$ from
the polytope $P$ by first, 
replacing, for each $w \in W$, the auxiliary variable
$y_w$ in $P$ by the sum of variables $y_{(t,\mathbf{j})}$
over the set of all identifiers $(t, \mathbf{j})$ of $y_w$;
and secondly, replacing, for every $v \in V$, the constraint $\gamma_v$,
by several copies of this constraint, one for every
identifier $(q, \mathbf{i})$ of $\gamma_v$. 
The obtained polytope lift $\bar{P}$ is clearly $k$-manageable
and it is easy to see that it recognizes the same property
of $L$-structures. 

To summarize, the constructions in Subsections~\ref{sec:rigid}, \ref{sec:supports}
and~\ref{sec:manageable} imply that
if a property of $L$-structures is recognised
by a family of $L$-symmetric LP lifts of size $s(n)$,
then the same property is
recognised by a family of $k(n)$-manageable
LP lifts, where 
$k(n) =
O(\log(s(n))/(\log(n)-\log\log(s(n))))$.

\subsection{From manageable polytopes to counting logic}\label{sec:translation}

We now put everything together in the proof
of Lemma~\ref{lem:secondhalf}.
For the sake of simplicity we give the proof
for the case when $L$ consists of a single binary symbol, that
is for the case of graphs. The general case is completely
analogous.

Let $P \subseteq \reals^{[n]^2} \times \reals^W$
be a graph-symmetric LP lift
of size $s$ recognising 
some property of graphs with $n$ vertices,
that is, a subset $A$ of $\{0,1 \}^{[n]^2}$.
We show that the same property of graphs is
definable by a
$\mathrm{C}^{k}$ formula, where $k =
O(\log(s)/(\log(n)-\log\log(s)))$.

Let $\widehat{P}$ be a rigid graph-symmetric LP lift
recognising $A$, as constructed in Subsection~\ref{sec:rigid}.
Recall that its size $s'$ is at most $s \log(s)$ where $s$ is the
size of $P$.  In particular, $s' \leq s^2$.

If $s > 2^{n/6}$, 
by a calculation analogous to~\eqref{eqn:stuff1}
at the beginning of Subsection~\ref{sec:supports}, we have
$n = O(\log(s)/(\log(n)-\log\log(s)))$.
Since every class
of graphs with $n$ vertices is definable in $\mathrm{C}^n$,
we complete the proof of the lemma in this case by taking
$k = n$.

If $s \leq 2^{n/6}$, then $s' \leq s^2 \leq 2^{n/3}$.
Hence, by Lemma~\ref{lem:supports},
for some $k =
  O(\log(s')/(\log(n)-\log\log(s')))$,
$\widehat{P}$ is $k$-supported, has at most $n^k$ auxiliary
variables, at most $n^k$ constraints, and
all its coefficients and constant terms can be
  encoded using at most $n^k$ bits.
  Moreover,
any such $k$ clearly satisfies 
$k =
  O(\log(s)/(\log(n)-\log\log(s)))$.
  
Let $\bar{P}$ be a $k$-manageable polytope lift
recognising $A$ (as described in Subsection~\ref{sec:manageable})
with 
the set of constraints indexed by
$Q \times [n]^k$, and
the set of auxiliary variables indexed by
$T \times [n]^k$.
Note that it follows from the construction of $\bar{P}$
that the number of elements in the sets $T$ and $Q$
is bounded, respectively, by the number of auxiliary variables and
the number of constraints in $\widehat{P}$.
Hence, $|Q|, |T| \leq n^k$.

Suppose now that we are given a graph $G$ with the set of
vertices $V$ of size $n$ and the set of edges $E$.
Intuitively, if we could fix a bijection between $[n]$
and $V$, we could then compute from $\bar{P}$ and $G$
a linear program $\bar{P}_G$ with the set of constraints $I$
and the set of variables $J$ as follows:
\begin{align*}
& I = \{\gamma_{(q,\mathbf{v})} : q \in Q, \mathbf{v} \in V^k\}, \\
& J = \{x_{vw} : v,w \in V\} \cup
\{y_{(t,\mathbf{v})} : t \in T, \mathbf{v} \in V^k\}.
\end{align*}

In order to decide whether $G$ has the property of interest
we would then check if the partial valuation:
$x_{vw} = 1$ if $(v,w) \in E$, and $x_{vw} = 0$ otherwise,
can be extended to a full solution.
This in turn can be easily done in logic using the following straightforward
consequence of the results in~\cite{Anderson:2015}.

\begin{lemma}\label{lem:fpc}
There exists an $\FPC$ formula $\phi$ which given a matrix
$\mathbf{A} \in \rationals^{I \times J}$ and a pair of vectors
$\mathbf{b} \in \rationals^I$, and
$\mathbf{a} \in \rationals^{J'}$, where $J' \subseteq J$,
decides if $\mathbf{a}$ can be extended to a solution
of the linear program $\mathbf{A} \mathbf{x} \leq \mathbf{b}$.
\end{lemma}

Our goal is to use Lemma~\ref{lem:propertymanageable}
to show that
the linear program $\bar{P}_G$
can be computed without fixing a bijection between $[n]$ and~$V$.
We define an $\FOC$-interpretation $\Psi$
which takes as input a
graph $G$ with $n$ vertices
and outputs, essentially, a relational encoding of
the linear program $\bar{P}_G$
together with the partial valuation discussed above.
More precisely, $\Psi$ outputs a matrix
$\mathbf{A} \in \rationals^{I \times J}$ and a pair of vectors
$\mathbf{b} \in \rationals^I$, and
$\mathbf{a} \in \rationals^{J'}$, where $J' \subseteq J$,
such that $\mathbf{a}$ can be extended to a solution
of $\mathbf{A} \mathbf{x} \leq \mathbf{b}$
if and only if $G$ has the property of interest.
To encode the fact that $J' \subseteq J$
we introduce an extra binary relation symbol $F$ of type
$\bar{J}' \times \bar{J}$ for an injective function from
the index set $J'$ to the index set $J$. Hence,
$\Psi$ outputs a structure over the vocabulary
$L_{\mathrm{vec},2} \disjointunion L_{\mathrm{vec},1}
\disjointunion L_{\mathrm{vec},1} \disjointunion \{F\}$
modified, for simplicity, in such a way that there is a single
sort symbol $\bar{B}$ for a domain of bit positions.

Given a graph $G$ with $n$ vertices
the $\FOC$-interpretation $\Psi$ has
access to the domain $V$ of the graph,
and the naturally ordered number domain $\{0, \ldots, n\}$.
To represent 
the bit encodings of the numerical coefficients we use
tuples from $[n]^k \subseteq \{0, \ldots, n\}^k$.
Let $o : [n]^k \rightarrow \{0, 1, \ldots, n^k-1\}$
be the order-preserving bijection from 
the set $[n]^k$ ordered lexicographically
to the set $\{0, 1, \ldots, n^k-1\}$ with the
natural order. For any $\svector \in [n]^k$,
by $[\svector]$ we denote the natural number~$o(\svector)$.
Tuples from $[n]^k \subseteq \{0, \ldots, n\}^k$
are also used to represent
elements of
$Q$ and $T$.
In order to do so let us fix an injective function $f$ from $Q$ to $[n]^k$,
and an injective function $g$ from $T$ to $[n]^k$.
The linear program
$\mathbf{A} \mathbf{x} \leq \mathbf{b}$
in the output of $\Psi$ has
the set of constraints indexed by $[n]^k \times V^k$
and the set of variables indexed by $V^2 \cup [n]^k \times V^k$.
Once restricted to the constraints indexed by $f(Q) \times V^k$
and the variables indexed by $V^2 \cup g(T) \times V^k$
it is exactly the linear program $\bar{P}_G$.
All the other coefficients and constant terms
in $\mathbf{A} \mathbf{x} \leq \mathbf{b}$
are set to
$0 = (-1)^0 \ 0/1$.

Consider tuples of the form
$(\zvector_1, \zvector_2, \zvector_3, \rho)$,
where $\zvector_1,\zvector_2, \zvector_3 \in [n]^k$, and
$\rho$ is a quantifier-free formula defining an equality
type of $2k$-tuples.
By $T^y_d$ let us denote the set of all tuples of this form
which satisfy one of the following conditions: 
\begin{itemize}\itemsep=0pt
\item $\zvector_1 \not\in f(Q)$ or
$\zvector_2 \not\in g(T)$, 
and $[\zvector_3] = 0$,
\item $\zvector_1 \in f(Q)$ and
$\zvector_2 \in g(T)$, and if $f^{-1}(\zvector_1) = q$,
$g^{-1}(\zvector_2) = t$,
then for every $\svector_1, \svector_2 \in [n]^k$ such that
the equality type of $(\svector_1,\svector_2)$ is $\rho$,
the position $[\zvector_3]$
in the binary encoding of
the denominator of the coefficient of the variable indexed by
$(t,\svector_2)$ in the constraint indexed by $(q, \svector_1)$
in $\bar{P}$ carries the $1$-bit.
\end{itemize}
It follows from Lemma~\ref{lem:propertymanageable}
that the set $T^y_d$ carries all information
about the denominators of the coefficients of the auxiliary variables
in $\bar{P}$.

Similarly, we define sets $T^{y}_s$, $T^{y}_n$,
$T^{x}_s$, $T^{x}_n$, $T^{x}_d$, and
$C_s$, $C_n$, $C_d$
to carry all the information about
the signs and the bits of the numerators
and the denominators of: the coefficients of the auxiliary variables, 
the coefficients of the
variables in $\{x_{ij}\}_{1 \leq i,j \leq n}$, and
the constant terms, respectively.

Given a graph $G$ with the set of
vertices $V$ of size $n$ and the set of edges $E$
the interpretation $\Psi$ does the following:
\begin{enumerate}\itemsep=0pt
\item defines the domain of $\bar{I}$ as $[n]^k \times V^k$,
the domain of $\bar{J}$ as $V^2 \cup [n]^k \times V^k$,
the domain of $\bar{J}'$ as $V^2$, and the domain of $\bar{B}$
as $[n]^k$,
where $V$ is the domain of $G$, and
$[n] \subseteq \{0, \ldots, n\}$ is a subset of the number domain,
\item defines the relation $\leq$ for the
linear order on $\bar{B}$ as the lexicographic order with respect to
the natural order of the number domain,
\item defines the relation $F$ of type $\bar{J}' \times \bar{J}$
as the equality relation on $V^2$,
\item defines the ternary relation
$P^\mathbf{A}_d$ of type $\bar{I} \times \bar{J} \times \bar{B}$
for encoding
the denominators of the entries of the matrix $\mathbf{A}$
as a union of two relations.
The first is defined as a subset of
$ ([n]^k \times V^k) \times
([n]^k \times V^k) \times [n]^k $ consisting of tuples
$(\mathbf{s}_1,\mathbf{v}_1, \mathbf{s}_2,\mathbf{v}_2, \mathbf{s}_3)$
for which there exists
$(\zvector_1, \zvector_2, \zvector_3, \rho)$
in $T^y_d$ such that
the tuple $(\mathbf{v}_1,\mathbf{v}_2)$ satisfies $\rho$, and
for every $i \in [3]$ it holds $\svector_i = \zvector_i$.
The second is defined as a subset of
$ ([n]^k \times V^k) \times
V^2 \times [n]^k$ consisting of tuples
$(\mathbf{s}_1,\mathbf{v}_1, v,w, \mathbf{s}_2)$ for which
there exists
$(\zvector_1, \zvector_2, \rho)$
in $T^x_d$ such that 
the tuple $(\mathbf{v}_1,v,w)$ satisfies $\rho$ and
$\svector_1 = \zvector_1$, and $\svector_2 = \zvector_2$,
\item defines the relations $P^\mathbf{A}_s$, $P^\mathbf{A}_n$,
$P^\mathbf{b}_s$, $P^\mathbf{b}_n$, $P^\mathbf{b}_d$ in a similar
way as $P^\mathbf{A}_n$,
\item defines the binary relations $P^\mathbf{a}_s$,
$P^\mathbf{a}_n$, $P^\mathbf{a}_d$ of type
$\bar{J}' \times \bar{B}$ for encoding 
the entries of the vector $\mathbf{a}$ 
in the following way:
the entries $((v,w),\svector)$ of $P^\mathbf{a}_s$,
$P^\mathbf{a}_n$, $P^\mathbf{a}_d$ are defined to encode
$1 = (-1)^0 \ 1/1$ or $0 = (-1)^0 \ 0/1$ depending
on whether $(v,w) \in E$ or not.
%
\end{enumerate}
Note that by existential quantification over
the sets $T^y_d$ and $T^x_d$ we really mean a disjunction.
And by $\svector_i = \zvector_i$ we mean the $2$-variable
$\FO$-formula of size $O(kn)$ which, for every
$j \in [k]$, says that the
$j$-th component $s_{i,j}$
of $\svector_i$ is the $z_{i,j}$-th component of $[n]$,
using the order on the number domain.
Observe also that $\Psi$, as described, is not rigorously an
$\FOC$-interpretation,
but it is not difficult to see that it can be easily turned into such.

The interpretation $\Psi$
has $O(k)$ variables. Its size is polynomial in $n^k$, in $k$, and in
the number of equality types of $2k$ tuples, that is,
polynomial in $n^k$, $k$, and $(2k)^{2k}$.
Since in our case $k = O(n)$,
the size of $\Psi$ is simply $n^{O(k)}$.

Now by composing $\Psi$ with the
$\FPC$ formula $\phi$ from Lemma~\ref{lem:fpc}
we obtain an $\FPC$ formula $\psi$ which given a
graph $G$ with $n$ vertices
decides if $G$ has the property of interest.
The formula $\psi$ 
has $l = O(k)$ variables and size 
$n^{O(k)}$.
We translate it into a formula $\theta$ of $\CLogic^{2l}$ such
that $\psi$ is equivalent to $\theta$ on all structures of size at
most $n$ and $\theta$ is of size polynomial in the size
of $\psi$, in $l$, and in $n^l$ (cf.\ Subsection~\ref{sec:logic}).
Hence, in terms of $k$ and $n$,
the formula $\theta$ has $O(k)$ variables and size
$n^{O(k)}$.

We have therefore shown that 
a property of graphs with $n$ vertices recognized by a
graph-symmetric polytope lift of size $s$
is defined by a $\mathrm{C}^{k}$ formula, where $k =
O(\log(s)/(\log(n)-\log\log(s)))$. Moreover,
if $s$ is at most weakly exponential, then for some positive
real $\epsilon$ we have
$k = O(\log(s)/(\log(n)-\log\log(s))) = 
O(\log(s)/(\epsilon \log(n))) = O(\log(s)/\log(n))$.
Hence, in this case the size of $\theta$
is $n^{O(k)} = s^{O(1)}$.
This finishes the proof of Lemma~\ref{lem:secondhalf}
and this section.

%% file: section-5-results-and-apps.tex
\section{Results and Applications} \label{sec:resultsandapps}

In this section we develop the main consequences of our results. We
start by establishing the main theorem of the paper, which
characterizes the expressive power of symmetric linear programs.  We
continue with the applications to upper and lower bounds. And end with
the observation that for random graphs over appropriate distributions
symmetric LP lifts are as powerful as general Boolean circuits.

\subsection{Equivalence of Models}

If $\mathscr{C}$ is a class of finite $L$-structures of some single-sorted
vocabulary $L$, and $n$ is a positive integer, we
write~$\mathscr{C}_n$ for the set of all structures in
$\mathscr{C}$ of cardinality $n$. We write $\symcirc_{\mathscr{C}}(n)$
for the size of a smallest $L$-symmetric Boolean circuit that
recognizes $\mathscr{C}_n$, and $\symlp_{\mathscr{C}}(n)$ for the
size of a smallest $L$-symmetric LP lift that recognizes
$\mathscr{C}_n$. Similarly, we write $\vars_{\mathscr{C}}(n)$ for
the \emph{counting-width} of~$\mathscr{C}_n$, i.e., the smallest
number of variables~$k$ of a $\CLogic^k$-formula that
defines~$\mathscr{C}_n$ on~$L$-structures of cardinality $n$, and
$\sizevars_{\mathscr{C}}(n)$ for the \emph{counting size-width} of
$\mathscr{C}_n$, i.e., the smallest $k$ such that there is a
$\CLogic^k$-formula of size at most $n^k$ that defines
$\mathscr{C}_n$ on~$L$-structures of cardinality $n$. 

\begin{theorem} \label{thm:main}
Let $\mathscr{C}$ be a class of finite $L$-structures of some
vocabulary~$L$. If $\symlp_{\mathscr{C}}(n)$ is at most weakly exponential, then
\begin{enumerate} \itemsep=0pt
  \item
  $\symcirc_{\mathscr{C}}(n)^{\Omega(1)} \leq \symlp_{\mathscr{C}}(n) \leq
\symcirc_{\mathscr{C}}(n)^{O(1)}$,
\item
  $\Omega(\sizevars_{\mathscr{C}}(n)) \leq
  \log(\symlp_{\mathscr{C}}(n))/\log(n) \leq
  O(\sizevars_{\mathscr{C}}(n))$.
\end{enumerate}
\end{theorem}

\begin{proof}
The upper bound in \emph{1} is a direct consequence of
Lemma~\ref{lem:circuitstolps}, and this holds without any assumption
on the growth rate of $\symlp_{\mathscr{C}}(n)$. The lower bound in
\emph{2} follows from Lemma~\ref{lem:secondhalf}: Write~$s =
\symlp_{\mathscr{C}}(n)$ and choose $k =
c\log(s)/(\log(n)-\log\log(s))$ for a large $c$ to be specified later.
Since $\symlp_{\mathscr{C}}(n)$ is at most weakly exponential we have
$s \leq 2^{n^{1-\epsilon}}$ for some $\epsilon>0$ and large enough
$n$. Hence $k = O(\log(s)/\log(n))$ with the hidden constant in the
big-oh notation dependent on $\epsilon$ as in $1/\epsilon$. Now, for
the appropriate constant in the big-oh in~$k = O(\log(s)/\log(n))$,
Lemma~\ref{lem:secondhalf} says that there is a $\CLogic^k$-formula
that defines $\mathscr{C}$ and has size polynomial in $s$, since again
$\symlp_{\mathscr{C}}(n)$ is at most weakly exponential. If the
constant in the big-oh in $k = O(\log(s)/\log(n))$ is chosen big
enough, we get that the size polynomial in $s$ is bounded even by
$n^k$, so $\sizevars_{\mathscr{C}}(n) = O(\log(s)/\log(n))$ as was to
be proved.  In turn, these two imply the lower bound in \emph{1} and
the upper bound in~\emph{2} through the well-known
relationship~$\symcirc_{\mathscr{C}}(n) \leq
n^{O(\sizevars_{\mathscr{C}}(n))}$ (see \cite{OttoHabilitationBook}).
\end{proof}


\subsection{Upper and Lower Bounds} \label{sec:upperandlower}

By Theorem~\ref{thm:main}, any class of graphs of unbounded counting
width cannot be recognized by polynomial-size symmetric LP lifts. More
strongly, in combination with the strongest known lower bounds on
counting width, Theorem~\ref{thm:main} gives weakly exponential lower
bounds of the type $2^{\Omega(n^{1-\epsilon})}$. We show that the
strongest forms of Lemmas~\ref{lem:circuitstolps}
and~\ref{lem:secondhalf} give even larger lower bounds.

\paragraph{Lower bounds on symmetric lifts and circuits}

In the sequel, let 3-XOR refer to the constraint satisfaction problem
of deciding whether a system of 3-variable parity constraints on
$\{0,1\}$-valued variables is satisfiable, and let~3-SAT refer to the
satisfiability problem for~3-CNF formulas.  In both cases, an instance
is presented as a finite structure that encodes the incidence
structure of the constraints: the domain is the disjoint union of the
set of variables and the set of constraints, and the relations carry
one monadic relation for each type of constraint that indicates which
constraints are of that type, and three binary relations that indicate
the three variables that participate in each constraint. Note that the
instances for these problems are not plain graphs but graphs with
coloured vertices and edges.

\begin{theorem} \label{thm:threecoloring}
  Every graph-symmetric LP lift or Boolean threshold circuit that
  recognizes the class of Hamiltonian graphs with $n$ vertices, or the
  class of 3-colourable graphs with $n$ vertices, or the class of
  satisfiable 3-SAT instances with $n$ variables, or the class of
  satisfiable 3-XOR instances with $n$ variables, has size
  $2^{\Omega(n)}$. Moreover, for 3-colouring, 3-SAT, and 3-XOR, the
  lower bound holds even on the class of instances with $O(n)$ edges,
  $O(n)$ clauses, and $O(n)$ constraints, respectively.
\end{theorem}

Before we enter the proof let us note that these $2^{\Omega(n)}$ lower
bounds for 3-colouring, 3-XOR and 3-SAT are optimal up to the
multiplicative constant in the exponent. We discuss this later in this
section.  Now we turn to the proof of
Theorem~\ref{thm:threecoloring}. First we handle 3-colourability, then
Hamiltonicity. As intermediate steps towards both we do 3-XOR and
3-SAT.

By Lemma~\ref{lem:secondhalf}, for obtaining the lower bound for LP
lifts it suffices to show that any~$\CLogic^k$-sentence that defines
the class of $n$-vertex 3-colourable graphs has $k = \Omega(n)$:
indeed, whenever~$s \leq 2^{n/d}$, we have
\begin{equation}
\log(s)/(\log(n)-\log\log(s)) \leq n/(d\log(d)). \label{eqn:boundd}
\end{equation}
By Lemma~\ref{lem:circuitstolps}, the claim then follows for Boolean
threshold circuits.  A result from the literature that is quite close
to the $k = \Omega(n)$ that we need can be found in
Section~4.2 in~\cite{Dawar1998}, but the analysis in there gives~$k =
\Omega(\sqrt{n})$, and not $k = \Omega(n)$. While it should be
possible to modify the construction in~\cite{Dawar1998} to get what we
need, we refer to a more recent construction that achieves what we
want for the problems 3-XOR and 3-SAT, and then proceed by
reduction. These intermediate steps will also be useful when we
discuss Hamiltonicity.

\begin{theorem}[see Theorem~3.7 and 3.8 in \cite{AtseriasDawar2018}
  and Lemmas~22 and~23 in \cite{DawarWang2017}]
  \label{thm:atseriasdawar}
  There exist $c,d > 0$ such that, for every $k$ and every
  sufficiently large $n$, every $\CLogic^k$-sentence that separates
  the class of satisfiable 3-XOR (resp. 3-SAT) instances with~$n$
  variables and~$cn$ constraints from the class of unsatisfiable
  ones~has~$k \geq dn$.
\end{theorem}

Neither~\cite{AtseriasDawar2018} nor~\cite{DawarWang2017} state the
linear bound $cn$ on the number of constraints, but it easily follow
from both proofs. Concretely, it follows from Lemma~3.3
in~\cite{AtseriasDawar2018}, in which the bound \emph{is} stated. Now
we proceed by reduction in order to get the same result for
3-colouring:

\begin{lemma} \label{lem:threecoloring}
  There exist $c,d > 0$ such that, for every $k$ and every
  sufficiently large $n$, every~$\CLogic^k$-sentence that separates
  the class of~3-colourable graphs with $n$ vertices and $cn$ edges
  from the class of non-3-colourable ones~has~$k \geq dn$.
\end{lemma}

\begin{proof}
In the textbook reduction from 3-SAT to the problem of
deciding whether a graph is 3-colourable (see, e.g.,
\cite{PapadimitriouBook}), the output graph has one gadget with two
vertices for each variable in the input formula, one gadget with six
vertices for each constraint of the input formula (the reduction in
\cite{PapadimitriouBook} uses only three vertices for each constraint
because the reduction starts at NAE-SAT; starting at 3-SAT we need six
vertices), and one special vertex. The edges are local to each gadget,
plus two edges from each variable gadget to the special vertex, and a
constant number of edges for each variable occurrence in the input
formula between the constraint gadget where the variable appears, and
the corresponding variable gadget. There are no other vertices or
edges in the graph. It is clear from the construction that if the
input formula has $n$ vertices and $m$ constraints, then this graph
has $O(n)+O(m)$ vertices and $O(n)+O(m)+O(m)$ edges, since each of the
$m$ constraints contributes three occurrences. And it is not difficult
to see that any~$\CLogic^k$ formula that separates the 3-colourable
graphs that are output by the reduction from the non-3-colourable ones
can be converted into a $\CLogic^{O(k)}$ formula that separate the
satisfiable 3-SAT instances that are input to the reduction from the
unsatisfiable ones. Another way to see this is by noting that the
reduction is definable by a uniform quantifer-free interpretation
(without the need for any ordering, or parameters, on the input
structure), from which the claim on $\CLogic^{O(k)}$-definability
follows from the closure of the logic under quantifier-free
interpretations (see Lemma~2.1 in \cite{AtseriasDawar2018}). Now the
claim follows from Theorem~\ref{thm:atseriasdawar}.
\end{proof}

For Hamiltonicity we follow the same path. A well-known result of
Dahlhaus \cite{Dahlhaus1983} gives a first-order definable reduction
from 3-SAT to Hamiltonicity that does not require any linear order on
the input. However, that reduction is quadratic and would only achieve
a lower bound of the form $k = \Omega(\sqrt{n})$ for $n$-vertex
graphs. We work out a linear reduction:

\begin{lemma} \label{lem:hamiltonicity}
There exists $d > 0$ such that, for every $k$ and every sufficiently
large $n$, every $\CLogic^k$-formula that defines the class of
Hamiltonian graphs with $n$ vertices has $k \geq dn$.
\end{lemma}

\begin{proof}
In this case we need a minor modification of the textbook reduction
from~3-SAT to the problem of deciding whether a graph is Hamiltonian
in, e.g., \cite{PapadimitriouBook}. As in the~3-colouring case, the output
graph of the textbook reduction has one constant-size gadget for each
constraint, and also one constant-size gadget for each variable. There
is also one constant-size gadget for each variable-occurrence in the
input 3-CNF formula, and three additional special vertices. This
defines all the vertices of the graph. See Figure~9.7
of~\cite{PapadimitriouBook}. Since each clause contributes three variable
occurrences, the total number of vertices is $O(n)+O(m)+O(m)$, where
$n$ is the number of variables, and $m$ is the number of clauses of
the input formula. However, for defining the edges of this reduction,
one needs a linear order on the variables of the input formula for
creating the path of variable gadgets in Figure~9.7
of~\cite{PapadimitriouBook}. Such a linear order is not available in our
encoding (and cannot be available as otherwise the proof of
Lemma~\ref{lem:threecoloring} breaks down). Here is where we modify
the construction. Instead of aligning the variable gadgets in the form
of a path, we arrange them in the form of a big clique as in
Figure~\ref{fig:standardmodified}.

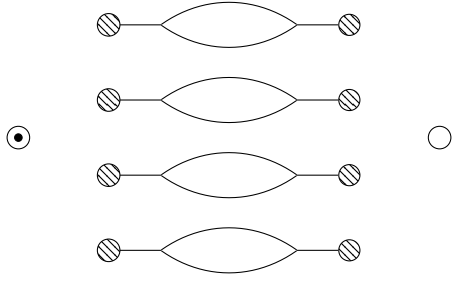
\begin{figure}
  \begin{center}
    \begin{tikzpicture}
      \usetikzlibrary{patterns}
      \draw (2.8cm, 2.5 cm) circle (0.15cm);
      \draw (-2.8cm, 2.5 cm) circle (0.15cm);
      \filldraw (-2.8 cm, 2.5 cm) circle (0.05cm);
      \foreach \y in {1,2,3,4} {
        \draw[pattern=north west lines] (-1.6cm, \y cm) circle (0.15cm);
        \draw (-1.45cm, \y cm) -- (-0.9cm, \y cm);
        \draw [domain=60:120] plot ({1.8*cos(\x)}, {\y-1.90+2.2*sin(\x)});
        \draw [domain=-60:-120] plot ({1.8*cos(\x)}, {\y+1.90+2.2*sin(\x)});
        \draw (1.45cm, \y cm) -- (+0.9cm, \y cm);
        \draw[pattern=north west lines] (1.6cm, \y cm) circle (4pt);
      }
    \end{tikzpicture}
  \end{center}
  \caption{The variable gadgets arranged in the form of a big
    (unordered) clique. All striped vertices are connected by a big
    clique. The rightmost and leftmost non-striped vertices denote the
    entry and exit points, respectively, and they are also connected
    to all the striped vertices. These gadgets are connected with the
    rest of the graph as in Figure~9.7 of~\cite{PapadimitriouBook}.}
  \label{fig:standardmodified}
\end{figure}

This modification does not introduce any new vertices and does not
depend on any given ordering of the variables in the input formula.
The analysis that proves that the reduction is correct is the same as
in~\cite{PapadimitriouBook}.  As in the proof of
Lemma~\ref{lem:threecoloring}, it is not hard to see that this
transformation is definable by a uniform quantifier-free
interpretation, still without parameters or linear orders, and the
result follows again from Theorem~\ref{thm:atseriasdawar} and the
closure of the logic under quantifier-free interpretations.
\end{proof}

\begin{proof}[Proof of Theorem~\ref{thm:threecoloring}]
  Apply Theorem~\ref{thm:atseriasdawar} and
  Lemmas~\ref{lem:threecoloring} and~\ref{lem:hamiltonicity} to
  Lemma~\ref{lem:secondhalf}, for LP lifts, and then to
  Lemma~\ref{lem:circuitstolps}, for circuits, and in both cases use
  equation~\eqref{eqn:boundd}.
\end{proof}

\paragraph{Lower Bound on the TSP Polytope}
As stated in the introduction, Yannakakis proved that the travelling
salesman polytope does not have subexponential symmetric LP
lifts. Here we show that this same lower bound follows from
Theorem~\ref{thm:threecoloring}. In the next section we will see that
the same type of argument cannot work for the matching polytope.

If $G = (V,E)$ is a graph with $V = [n]$, the incidence vector of $G$
is $\mathbf{x}^G = ( x^G_{ij} : i,j \in [n] )$, the vector in
$\reals^{[n]^2}$ defined by $x^G_{ij} = 1$ if $(i,j)$ is an edge of
$G$, and $x^G_{ij} = 0$ if $(i,j)$ is not an edge of $G$. Let
$\mathrm{TSP}_n$ denote the convex hull of all the vectors of the form
$\mathbf{x}^C$, where $C$ is a Hamilton cycle of the complete graph
$\mathrm{K}_n$ on $n$ vertices. We want to show that every
graph-symmetric~LP lift that has $\mathrm{TSP}_n$ as shadow must be of
exponential size as a consequence of Theorem~\ref{thm:threecoloring}.

\begin{theorem}[Theorem 2 in \cite{Yannakakis1991}] \label{thm:yannakakis}
  Every graph-symmetric LP lift that has $\mathrm{TSP}_n$ as shadow
  has size $2^{\Omega(n)}$.
\end{theorem}

\begin{proof}
By Theorem~\ref{thm:threecoloring}, it suffices to show that if
$\mathrm{TSP}_n$ were the shadow of a subexponential-size
graph-symmetric~LP lift, then there would be a subexponential-size
graph-symmetric~LP lift that separates the $\xvector^G$ for which $G$ is
Hamiltonian from the $\xvector^G$ for which $G$ is not Hamiltonian. Assume
then that $P$ were symmetric LP lift of size $2^{o(n)}$ whose shadow
is $\mathrm{TSP}_n$; let us say that its principal variables are
$\mathbf{y} = (y_{ij} : i,j \in [n])$ and that its auxiliary variables
are $\mathbf{z} = (z_k : 1 \leq k \leq p(n))$. Consider the following
LP on the $x_{ij}$, $y_{ij}$ and $z_k$ variables:
\begin{equation}
\begin{array}{lll}
0 \leq y_{ij} \leq x_{ij} & & \text{ for each } i,j \in [n] \\
(\mathbf{y},\mathbf{z}) \in P. & & 
\end{array}
\end{equation}
Since by assumption $P$ is graph-symmetric, $Q$ is also
graph-symmetric.  Hence, it suffices to show that the projection of
$Q$ on the $x$-variables separates Hamiltonian graphs from
non-Hamiltonian graphs.

It is clear that if $G$ is a graph that contains a Hamilton cycle $C$,
then $\mathbf{x}^G$ is in the projection of~$Q$ on the $x$-variables:
choose $\mathbf{y} = \mathbf{x}^C$, and let $\mathbf{z}$ witness its
membership in~$P$. Conversely, assume that $G$ is a graph and
that~$\mathbf{x}^G$ is in the projection of $Q$ on
the~$x$-variables. Then there exists~$\mathbf{y}^* = (y^*_{ij})_{ij}$
that is in the TSP polytope and satisfies the inequalities~$0 \leq
y^*_{ij} \leq x^G_{ij}$ for each~$i,j \in [n]$. This means that the
support of~$\mathbf{y}^*$ defines a subgraph of~$G$, and at the same
time that~$\mathbf{y}^*$ is a convex combination of Hamilton cycles
of~$\mathrm{K}_n$. Let $\mathbf{y}^C$ be one of the vectors in this
convex combination, where $C$ is a Hamilton cycle of
$\mathrm{K}_n$. The support of $\mathbf{y}^C$ is of course included in
the support of $\mathbf{y}^*$, which means that $C$ is also a subgraph
of~$G$. So~$G$ contains a Hamilton cycle and is thus Hamiltonian.
\end{proof}

\paragraph{Upper bounds}

Let us start with the simple observation that the lower bounds of type
$2^{\Omega(n)}$ on 3-colouring, 3-SAT and 3-XOR are optimal: all three
cases can be solved by symmetric Boolean threshold circuits of size
$2^{O(n)}$, and hence by symmetric LP lifts of size~$2^{O(n)}$ by
Lemma~\ref{lem:circuitstolps}. Here $n$ is the number of vertices or
variables, respectively. This follows from the fact that all three
problems are definable in the existential fragment of monadic
second-order logic (i.e., monadic NP), which on structures of size
$n$, straightforwardly translate into symmetric Boolean circuits of
size~$2^{O(n)}$. For Hamiltonicity, the straightforward symmetric
upper bound is only $2^{O(n\log n)}$ on $n$-vertex graphs. It looks
plausible that Bellman \cite{Bellman1962} and Held-Karp
\cite{HeldKarp1961} dynamic programming $O(n^2 2^n)$ algorithm
\cite{Bellman1962,HeldKarp1961} could be implemented in a symmetric
Boolean circuit, but we are not aware of a reference where this has
been worked out. What is known is that Hamiltonicity is not definable
in (full) monadic second-order logic (see \cite{EbbinghausFlumBook}).

The list of problems for which we proved a lower bound in
Theorem~\ref{thm:threecoloring} includes one that is decidable in
polynomial time, i.e., 3-XOR. This means that any polynomial-size
family of LP lifts or threshold circuits that recognizes 3-XOR must
\emph{a fortiori} be asymmetric. On the other hand,
Theorem~\ref{thm:main} says that any problem that is definable in
$\FPC$ has polynomial-size symmetric LP-lifts.  This includes graph
planarity \cite{Grohe1998}, any polynomial-time decidable property of
graphs that exclude some minor~\cite{Grohe2011}, matrix singularity
over rationals \cite{BlassGurevichShelah1999}, solving systems of
linear equations over rationals \cite{Holm2011}, and many others. By
the results in~\cite{Anderson:2015}, it also includes the problem of
deciding whether a (general, not necessarily bipartite) graph contains
a perfect matching. The family can even be taken to be polynomial-time
uniform by applying the construction of Lemma~\ref{lem:circuitstolps}
to the ``easy half'' of the equivalence between FPC and
polynomial-size symmetric threshold circuits in
\cite{AndersonDawar2017}.

\begin{corollary} \label{cor:matching}
There is a (polynomial-time uniform) family of graph-symmetric~LP
lifts of polynomial size that recognizes the class of graphs that have
a perfect matching.
\end{corollary}

This should be contrasted with the fact, proved by Yannakakis, that
any symmetric LP lift of the perfect matching polytope $\mathrm{PM}_n$
has size $2^{\Omega(n)}$. Here, $\mathrm{PM}_n$ is defined as the
convex hull of all vectors of the form $\xvector^M$, where $M$ is the
edge set of a perfect matching of the complete graph $\mathrm{K}_{2n}$
on $2n$ vertices. Capturing $\mathrm{PM}_n$ by an LP lift or
recognizing the class of graphs that have a perfect matching by an LP
lift are different tasks. Both objects could be used for deciding
whether a given graph has a perfect matching, but capturing
$\mathrm{PM}_n$ has a demanding structural requirement that has no
analogue in the other task. The upper bound of
Corollary~\ref{cor:matching} also means that the argument that was
used for deriving lower bounds for $\mathrm{TSP}_n$ in the proof of
Theorem~\ref{thm:yannakakis} cannot be adapted to
$\mathrm{PM}_n$. Indeed, we do not know whether there is any route at
all for deriving lower bounds for $\mathrm{PM}_n$ via our results.

In view of Corollary~\ref{cor:matching} one may wonder whether the
convex hull of the incidence vectors of graphs that have a perfect
matching has a small symmetric LP lift. This, however, is easily seen
to not be the case: if it had, then its intersection with the
halfplane $\sum_{i,j \in [2n]: i\not=j} x_{ij} = 2n$ would be a small symmetric
LP lift for $\mathrm{PM}_n$.

\subsection{Problems on Erd\H{o}s-R\'enyi Random Graphs} \label{subsec:averagecase}

Let $\mathscr{G}(n,p)$ denote the Erd\H{o}s-R\'enyi distribution on
$n$-vertex labelled graphs with edge probability $p$. We write $G \sim
\mathscr{G}(n,p)$ to mean that $G$ is a random graph distributed as
in~$\mathscr{G}(n,p)$. In this section we argue that, for average-case
problems with respect to the uniform
distribution~$\mathscr{G}(n,1/2)$, as well as for the type of problems
that ask to distinguish $\mathscr{G}(n,1/2)$ from some other
distribution, polynomial-size symmetric LPs are as powerful as
arbitrary not necessarily symmetric Boolean circuits. For average-case
problems, this is indeed a direct consequence of our main result and
the following well-known fact in descriptive complexity theory:

\begin{theorem}[Corollary 4.8 in \cite{HellaKolaitisLuosto1996}] \label{HellaKolaitisLuosto}
For every polynomial-time decidable class of graphs~$\mathscr{C}$
there is an $\FPC$-definable class of graphs $\mathscr{C}'$ for which
the probability that a random graph~$G \sim \mathscr{G}(n,1/2)$ falls
in the symmetric difference $\mathscr{C} \symdiff \mathscr{C}'$ is
$o(1)$.
\end{theorem}

The point of Theorem~\ref{HellaKolaitisLuosto} is that the $\FPC$
formula that defines $\mathscr{C}'$ does not require any order on the
input graph, hence our Theorem~\ref{thm:main}
applies. Theorem~\ref{HellaKolaitisLuosto} is indeed a consequence of
the Immerman-Vardi Theorem \cite{Immerman1987,Vardi1982} and the fact
that a linear order is, asymptotically almost surely on
$\mathscr{G}(n,1/2)$, definable in FPC. We return to this later.  For
the rest of this section let us focus our discussion on the problem of
distinguishing $\mathscr{G}(n,1/2)$ from some other distribution of
random graphs, to which a direct application of
Theorem~\ref{HellaKolaitisLuosto} does not look possible.

We focus the discussion on the planted clique problem since it is one
of the best studied such problems, although it will be clear from the
discussion that the phenomenon is more general.
Let~$\mathscr{G}(n,p,k)$ denote the distribution that
results from drawing a random graph from~$\mathscr{G}(n,p)$ and then
\emph{planting} a random $k$-clique in it, i.e., adding the edges of a
$k$-clique on a uniformly chosen subset of $k$ vertices. Following
\cite{BarakSteurer2016Notes}, the planted clique problem, also known
as the \emph{hidden} clique problem, comes in three flavours.
Informally stated, these are:

\begin{itemize} \itemsep=0pt
\item search: given $G \sim \mathscr{G}(n,p,k)$,
find the planted clique,
\item refutation: given $G \sim \mathscr{G}(n,p)$,
certify that the clique number is less than $k$,
\item decision: given $G \sim \mathscr{G}(n,p)$ or $G \sim
  \mathscr{G}(n,p,k)$, determine  which is the case.
\end{itemize}

\noindent Formally, the decision version can be stated as follows. We
say that a class of graphs $\mathscr{C}$ solves the decision version
of the planted clique problem with parameters~$p = p(n)$ and~$k =
k(n)$ and advantage $\epsilon = \epsilon(n) > 0$ if for every large
enough~$n$ the following hold:
\begin{enumerate} \itemsep=0pt
\item if $G \sim \mathscr{G}(n,p)$, then $G$ is
in $\mathscr{C}$ with probability at least $1/2+\epsilon$,
\item if $G \sim
\mathscr{G}(n,p,k)$, then $G$ is in $\mathscr{C}$
with probability at most $1/2-\epsilon$. 
\end{enumerate}
Solvable in polynomial time (or in FPC, or by a family of LP-lifts,
etc.)  means that it is solvable by a class of graphs $\mathscr{C}$
that can be recognized in polynomial time (or in FPC, or by a family
of LP-lifts, etc.).  It should be clear that if the decision version
is hard, then the other two versions of the problem can only be
harder.

The planted clique problem is an easy-to-state average case problem
that has a long history. Its formulation is attributed to Jerrum
\cite{Jerrum1992} and Ku\v{c}era \cite{Kucera1995}, independently. In
the range~$k(n) = \omega(\sqrt{n\log n})$, a simple algorithm based on
degree sequences solves the search problem in polynomial time
\cite{Kucera1995}. In the range~$k(n) = \Omega(\sqrt{n})$, a spectral
based algorithm is known to solve the problem in polynomial time
\cite{AlonKrivelevichSudakov1998}, but spectral-free algorithms are
also known that run even in linear time
\cite{FeigeKrauthgamerpre2013}. For $k(n) = o(\sqrt{n})$ the status of
the problem is famously open, even for its decision variant. Some
lower bounds are known in restricted models, such as the original
lower bound for Markov chain methods in \cite{Jerrum1992}. Lower
bounds are also known for (symmetric) linear and semidefinite program
formulations of the problem. We discuss these next.

We formulate the clique problem on a graph $G = (V,E)$ as a non-linear
polynomial optimization problem: For each vertex~$v \in V$ introduce
one variable $y_v$ that stands for the indicator whether~$v$ belongs
to the clique. The clique number is the maximum of $\sum_{v \in V}
y_v$ subject to the constraints that $y_u y_v = 0$ for each non-edge
$(u,v) \not\in E$, and $y_v^2 - y_v = 0$ for each~$u \in V$. For the
refutation and the decision versions of the problem, it is more
natural to turn the objective function into a constraint $\sum_{v \in
  V} y_v \geq k$.  This gives a different quadratic program
feasibility problem for each graph $G$. Thinking of $G$ as given by
the $\{0,1\}$-vector~$(x_{u,v})_{u,v \in [n]}$ in the usual way, the
programs can be made uniform, i.e., a single quadratic program
feasibility problem serves all graphs with $V = [n]$:

\medskip
\begin{center}
\begin{tabular}{lll}
& $\sum_{v \in [n]} y_v \geq k$ & \\
& $y_uy_v \leq x_{u,v}$ & for $u,v \in [n]$ \\
& $y_v^2 - y_v = 0$ & for $v \in [n]$
\end{tabular}
\end{center}
\medskip

While this a hard-to-solve quadratic program, there are several
tractable relaxations that one can study. Two systematic methods for
generating such relaxations were introduced by Lov\'asz and Schrijver
in \cite{LovaszSchrijver1991}, and Sherali and Adams in
\cite{SheraliAdams1990}. Both cases start by
(naively)~\emph{linearizing} the quadratic program, i.e., replacing
each quadratic term~$y_u y_v$ in the constraints by a new variable
$y_{\{u,v\}}$. The result is a very weak symmetric LP relaxations with
auxiliary variables, but this is only the \emph{first level} of the~LS
and~SA hierarchies.  The successive levels of the hierarchies are
obtained by adding linear constraints that can be proved valid for the
convex hull of solutions of the quadratic program by iterated
applications of a simple rule of inference.

It can be seen that the levels of the hierarchies are always
graph-symmetric LPs, with the~$d$-th level having $n^{O(d)}$ auxiliary
variables and constraints, where~$n = |V|$. The strengths of the
relaxations converge to the quadratic program in the sense that the
polytopes they project to are tighter and tighter approximations of the
convex hull of solutions of the quadratic
program. Moreover, the exact convex hull is eventually reached no
later than the~$n$-th level.

In order to appreciate the symmetry of the LPs that define the
sucessive levels of the hierarchy, it useful to enter the details for
the definition of the SA hierarchy. The first step in producing
the~$d$-th level of the SA hierarchy is to formally multiply each
constraint of the quadratic program by an extended monomial of the
form $\prod_{w \in I} y_w \prod_{w \in J} (1-y_w)$ for~$I,J \subseteq
V$ with $I \cap J = \emptyset$ and $|I \cup J| \leq d-1$. The second
step is to expand out the formal expressions into sums of degree $d+1$
monomials, and multilinearize every monomial on the $y$-variables;
note that this is a valid step only under the constraint $y_v^2 - y_v
= 0$, hence for $\{0,1\}$-assignments. In the third step we introduce
a new variable $y_I$ for each $I \subseteq V$ with $|I| \leq d+1$, and
a new variable $y_{u,v,I}$ for each $u,v \in V$ and each $I \subseteq
V$ with $|I| \leq d-1$. The fourth and final step is to linearize:
replace each monomial $\prod_{w \in I} y_w$ by the new variable $y_I$,
and each monomial~$x_{u,v} \prod_{w \in I} y_w$ by the new variable
$y_{u,v,I}$. The link between the two types of variables is
established by setting $x_{u,v} = y_{u,v,\emptyset}$ and
$y_{\emptyset} = 1$. The symmetry of the resulting~LP is obvious since
the image~$I^{(\pi)}$ of a set $I \subseteq V$ by a permutation $\pi
\in \Sym_V$ has the same cardinality as~$I$.  The result is a
graph-symmetric LP with~$n^{O(d)}$ auxiliary variables and
constraints, whose shadow eventually captures the convex hull of
solutions of the quadratic program. Any level that approximates this
hull well enough, even on average on $\mathscr{G}(n,1/2)$, can be used
for solving the planted clique problem.

The strengths and limitations of the LS and SA hierarchies (and
beyond) for the planted clique problem have been the object of
intensive study. Starting with the work of Feige and Krauthgamer
\cite{FeigeKrauthgamer2003}, it is known that for each constant $d$,
the $d$-th level of the LS hierarchy cannot distinguish
between $\mathscr{G}(n,1/2)$ and
$\mathscr{G}(n,1/2,k(n))$ when $k(n) = o(\sqrt{n})$
with any significant advantage. Their lower bound is much stronger
than we stated it since it applies even to LS$^+$, the semi-definite
programming variant of the LS-hierarchy. For the~SA hierarchy, the
same result is attributed to the folklore, and more recent works
\cite{Baraketal2016,Hopkinsetal2018} have obtained analogous results
for the even stronger Lasserre/Sums-of-Squares~(SOS) hierarchies, also
proving that the problem stays hard for their constant $d$ level when
$k(n) = o(\sqrt{n})$.  Further, in certain contexts, it is possible to
prove that the Sherali-Adams hierarchy is \emph{optimal} among
symmetric~LP lifts of comparable size. This includes a model of LP
lifts for solving Boolean constraint satisfaction problems (see
Theorem~4.1 in \cite{ChanLeeRaghavendraSteurer2016}). This means that,
in such restricted contexts, proving lower bounds on the levels
of Sherali-Adams is enough for getting size lower bounds for \emph{any}
symmetric LP lift.

In view of such success in proving lower bounds on the size of 
symmetric LP lifts, starting with Yannakakis, and including the
discussion above on hierarchies for the planted clique problem, and
also given our own lower bounds from Section~\ref{sec:upperandlower}, the
following consequence of Theorem~\ref{thm:main} may come as a
surprise:

\begin{corollary} \label{cor:plantedclique}
If the planted clique problem with parameters $p = 1/2$ and $k = k(n)$
is solvable in polynomial time with advantage $\epsilon > 0$, then it
is also solvable by a (polynomial-time uniform) family of
polynomial-size graph-symmetric LP lifts with advantage $\epsilon -
o(1)$.
\end{corollary}

In the rest of this section we show how to derive the descriptive
complexity version of Corollary~\ref{cor:plantedclique}, from which
Corollary~\ref{cor:plantedclique} follows at once from
Theorem~\ref{thm:main}. The descriptive complexity variant relies on
the following well-known fact, which builds on the almost sure graph
canonization methods of \cite{BabaiErdosSelkow1980}:

\begin{theorem}[Theorem 4.6 in \cite{HellaKolaitisLuosto1996}] 
\label{HellaKolaitisLuosto2}
There is an $\FOC$-formula $\phi(x,y)$ such that, for all $\epsilon >
0$ and all sufficiently large $n$, if $G \sim \mathscr{G}(n,1/2)$, then
$\phi(x,y)$ defines a strict linear order on the vertices of $G$ with
probability at least $1-\epsilon$.
\end{theorem}

We note that Theorem~\ref{HellaKolaitisLuosto2} is one of the two
ingredients in the proof of Theorem~\ref{HellaKolaitisLuosto}; the
second one is the Immerman-Vardi Theorem. While our proof uses only
these two ingredients too, we do not see a direct way of getting it
from Theorem~\ref{HellaKolaitisLuosto}.

\begin{theorem} \label{thm:plantecliquefpc}
If the planted clique problem with parameters $p = 1/2$ and $k = k(n)$
is solvable in polynomial time with advantage $\epsilon > 0$, then it
is solvable in $\FPC$ with advantage~$\epsilon - o(1)$.
\end{theorem}

\begin{proof}
Suppose that $\mathscr{C}$ is a polynomial time decidable class of
graphs that solves the problem with advantage $\epsilon > 0$. Let
$\mathscr{C}'$ denote the class of all ordered expansions of graphs in
$\mathscr{C}$, i.e., the structures in $\mathscr{C}'$ are finite
structures over the vocabulary $L = \{ E, R \}$, where $E$ and $R$ are
binary relations symbols, and $E$ is interpreted as the edge relation
of some graph in $\mathscr{C}$, and $R$ is interpreted as a strict
linear order on its set of vertices. By the Immerman-Vardi Theorem,
there is an FP formula $\psi$ that defines $\mathscr{C'}$ on the class
of all ordered graphs. Let~$\phi(x,y)$ be the~$\FOC$ formula from
Theorem~\ref{HellaKolaitisLuosto2}, and let $\theta$ be the
conjunction of the following sentences:
\begin{align*}
& \forall x (\neg\phi(x,x)) \\
& \forall x \forall y \forall z (\neg\phi(x,y) \vee \neg\phi(y,z) \vee \phi(x,z)) \\
& \forall x \forall y (x=y \vee \phi(x,y) \vee \phi(y,x)) \\
& \psi[R/\phi].
\end{align*}
This sentence says that $\phi$ defines a strict linear order and
$\psi$ holds when each occurrence of~$R$ is replaced by $\phi$. This
is an FPC formula over the vocabulary of unordered graphs that defines
a class $\mathscr{D}$ of graphs. We claim that $\mathscr{D}$ solves
the problem with advantage $\epsilon-o(1)$.
 
By assumption and the fact that $\mathscr{C}'$ contains all ordered
expansions of graphs in $\mathscr{C}$ we have that the probability
that some and hence every ordered expansion of $G$ satisfies $\psi$ is
at least~$1/2+\epsilon$ when $G \sim \mathscr{G}(n,1/2)$, and at most
$1/2-\epsilon$ when $G \sim \mathscr{G}(n,1/2,k)$.
Now, if~$G \sim \mathscr{G}(n,1/2)$, then the probability that $\phi$
does not define a linear order is $o(1)$, and the probability that
some and hence every ordered expansion of $G$ satisfies $\psi$ is at
least~$1/2+\epsilon$, so the probability that $G$ satisfies $\theta$
is at least~$1/2+\epsilon-o(1)$. On the other hand, if $G \sim
\mathscr{G}(n,1/2,k)$, then the probability that some
and hence every ordered expansion of $G$ satisfies $\psi$ is at most
$1/2-\epsilon$, so the probability that $G$ satisfies~$\theta$ is even
smaller, and $1/2-\epsilon \leq 1/2-\epsilon + o(1)$. It follows that
$\mathscr{D}$ solves the problem with advantange~$\epsilon-o(1)$.
\end{proof}

%% file: section-6-concluding-remarks.tex
\section{Concluding Remarks} \label{sec:concludingremarks}

Our main result Theorem~\ref{thm:main} establishes a tight three-way
correspondence between symmetric Boolean threshold circuits, symmetric
LP lifts, and bounded-variable formulas of counting logic. We used
this to derive upper and lower bounds on the size of symmetric LPs and
symmetric Boolean threshold circuits. We also used it to bound the
asymmetric circuit complexity of the planted-clique problem by its
symmetric LP lift complexity, up to a polynomial factor.  There are several
directions for further investigation that are suggested by this work.

The first one concerns the problem of \emph{circuits} vs
\emph{formulas}.  Composing the first inequality in the first item of
Theorem~\ref{thm:main} with the second inequality in the second item
of Theorem~\ref{thm:main}, we get the result that, for every constant
$k$, every symmetric Boolean threshold circuit of size at most $n^k$
translates into an equivalent $\CLogic^{O(k)}$-formula of size
$n^{O(k)}$. This is a size-efficient translation from circuits into
formulas, albeit of different types: the source is a Boolean threshold
circuit and the target is a counting logic formula.  An explicit and
uniform such translation is given in~\cite{AndersonDawar2017}, and
similarly, AND-OR-NOT circuits can translate into families of $\mathrm{L}^k$-formulas, where $\mathrm{L}^k$ stands for the
$k$-variable fragment of first-order logic, without counting.

At first sight, the translation from circuits to formulas is
unexpected as the circuit value problem is P-complete, while the
formula value problem is in NC$^1$.  However, as we noted, these are
not Boolean formulas and the natural translation of formulas of
$\CLogic^k$ or $\mathrm{L}^k$ into circuits necessarily yields
circuits with high fan-out.  This also accords with the well-known
fact that the combined complexity of the logic $\mathrm{L}^k$ is
P-complete for each $k \geq 3$ \cite{Vardi1995}.  An intriguing
question at this point is: What happens when we start the translation
not from a circuit but from a symmetric Boolean (threshold) formula?
Does this correspond to a natural syntactic fragment of the logic
$\mathrm{L}^k$ (or $\CLogic^k$)?  Could such a translation shed light
on the descriptive complexity of NC$^1$? On the NC$^1$ vs P question?
Similar questions can be asked for symmetric bounded-depth threshold
circuits, and other classes of circuits. It is worth noting that both
Theorem~\ref{HellaKolaitisLuosto} and
Theorem~\ref{thm:plantecliquefpc} on the planted-clique problem scale
all the way down to uniform TC$^0$ and $\FOC$. This is so because
Theorem~\ref{HellaKolaitisLuosto2} gives a $\FOC$-formula and, in the
presence of a linear order on the input, $\FOC$ captures uniform
TC$^0$ (see Proposition~12.6 in \cite{ImmermanBook}). We view all this
as motivation for finding a symmetric LP model of symmetric
bounded-depth threshold circuits.

A different line of research that is suggested by our work relates to
the computational complexity of the linear programming feasibility
problem. It is well-known that this problem is P-complete under
logspace reductions and a question posed in~\cite{Anderson:2015},
where it was shown that this problem is in FPC, is whether it is
complete for FPC under (say) FOC reductions.  Our proof of
Lemma~\ref{lem:circuitstolps} gives one route towards answering this
question.  The construction in the proof of the lemma should yield a
first-order interpretation from the threshold circuit value problem to
the LP feasibility problem.  Since the known translations from FPC
into symmetric threshold circuits are first-order interpretations
themselves, the result should follow by composition. The
FPC-completeness of the~LP feasibility problem was also proved by
Pakusa (unpublished manuscript), independently of our work. Pakusa's
construction is very different from ours and raises the question
whether the construction in his proof could yield a different proof of
Lemma~\ref{lem:circuitstolps}. His construction is inspired by the
Sherali-Adams hierarchy of LP relaxations \cite{SheraliAdams1990}
(discussed in Section~\ref{sec:resultsandapps}) and could well provide
a more principled method than ours for constructing LP lifts with
useful properties.